\documentclass[lettersize,journal]{IEEEtran}
\usepackage{amsmath,amsfonts,bm}
\usepackage{amssymb}
\usepackage{mathtools}
\usepackage{amsthm}
\newcommand{\attcolor}{black}
\usepackage{algorithmic}
\usepackage[table]{xcolor} 
\definecolor{mypurple}{HTML}{E8EEDD}
\usepackage[ruled,vlined,linesnumbered]{algorithm2e}
\usepackage{array}
\usepackage[caption=false,font=normalsize,labelfont=sf,textfont=sf]{subfig}
\usepackage{textcomp}
\usepackage{stfloats}
\usepackage{url}
\usepackage{verbatim}
\usepackage{graphicx}
\usepackage{cite}
\usepackage{pifont}
\usepackage{wrapfig}
\usepackage{amsmath}
\usepackage{amssymb}
\usepackage{mathtools}
\usepackage{amsthm}
\newtheorem{definition}{Definition}
\newtheorem{assumption}{Assumption}
\newtheorem{remark}{Remark}
\newtheorem{theorem}{Theorem}
\newtheorem{corollary}{Corollary}
\newtheorem{lemma}{Lemma}
\usepackage{multirow}
\usepackage{makecell}
\allowdisplaybreaks[4]
\usepackage{booktabs}
\usepackage{color}
\usepackage{hyperref}
\usepackage[table,xcdraw]{xcolor}

\usepackage{bookmark}
\begin{document}

\title{Distributed Bilevel Optimization with Dual Pruning for Resource-limited Clients}

\author{Mingyi Li, \and Xiao Zhang, \and Ruisheng Zheng, \and Hongjian Shi, \and Yuan Yuan, \and Xiuzhen Cheng, \and and Dongxiao Yu
\thanks{Mingyi Li, Xiao Zhang, Hongjian Shi, Xiuzhen Cheng and Dongxiao Yu are with the School of Computer Science and Technology, Shandong University, Qingdao 266237, China.\protect ~Email: limee@mail.sdu.edu.cn; xiaozhang@sdu.edu.cn; hongjian@mail.sdu.edu.cn; xzcheng@sdu.edu.cn; dxyu@sdu.edu.cn.}
\thanks{Ruisheng Zheng is with the School of Cyber Science and Technology, Shandong University, Qingdao 266237, China.\protect ~Email: rex7@mail.sdu.edu.cn.}
\thanks{Yuan Yuan is with the School of Software $\&$ Joint SDU-NTU Centre for Artificial Intelligence Research (C-FAIR), Shandong University, Jinan, 250000, China.\protect ~Email: yyuan@sdu.edu.cn.}}

\markboth{Journal of \LaTeX\ Class Files,~Vol.~14, No.~8, August~2021}%
{Shell \MakeLowercase{\textit{et al.}}: A Sample Article Using IEEEtran.cls for IEEE Journals}

\IEEEpubid{0000--0000/00\$00.00~\copyright~2021 IEEE}

\maketitle

\begin{abstract}
With the development of large-scale models, traditional distributed bilevel optimization algorithms cannot be applied directly in low-resource clients. 
The key reason lies in the excessive computation involved in 
optimizing both the lower- and upper-level functions. 
Thus, we present the first resource-adaptive distributed bilevel optimization framework with a second-order free hypergradient estimator, which allows each client to optimize the submodels adapted to the available resources.
Due to the coupled influence of partial outer parameters $x$ and inner parameters $y$, it's challenging to theoretically analyze the upper bound regarding the globally averaged hypergradient for full model parameters. The error bound of inner parameter also needs to be reformulated since the local partial training.
The provable theorems show that both \textit{RABO} and \textit{RAFBO} can achieve an asymptotically optimal convergence rate of $\mathcal{O}(1/\sqrt{\mathcal{C}_x^*Q})$, which is dominated by the minimum coverage of the outer parameter $\mathcal{C}_x^*$. 
Extensive experiments on two different tasks 
demonstrate the effectiveness and computation efficiency of our proposed methods. 
\end{abstract}


\section{Introduction}
\IEEEPARstart{R}{ecently}, distributed bilevel optimization has arisen as a powerful tool to capture various modern machine learning problems~\cite{gao2023convergence, HeHHLWY24compress}.
For instance, FedCSS~\cite{li2024joint} and Bi-level~\cite{huang2022federated} formulate the distributed nodes selection and global model construction as a distribted bilevel optimization problem. Bi-AC~\cite{zhang2020bi} proposes a bi-level reinforcement learning method while unequal agents can act simultaneously and distributedly. 
And distributed meta learning~\cite{fallah2020personalized, chen2018federated} is another important application of bilevel optimization to learn better model initializations. 
With the prosperity of LLMs~\cite{minaee2024large,zhang2024urban}, FedBioT~\cite{wu2024fedbiot} formulates the privacy-concerned 
fine-tuning as a bi-level optimization problem to mitigate the negative effect of data discrepancy. 
Generally, the optimization objectives of classic distributed bilevel problems ~\cite{conf/icml/fednest, conf/icml/fedmbo} can be  formalized as follows: 
\begin{equation}
\scalebox{0.9}{$
\begin{aligned}\label{objective}
&\min_{x\in\mathbb{R}^{d_1}} \Phi(x)=f(x, y^*(x)) : = 
\frac{1}{n} \sum_{i=1}^n f_i(x, y^*(x))
\\& \;\;\mbox{s.t.} \quad y^*(x)\in \arg\min_{y\in\mathbb{R}^{d_2}} g(x,y):=
\frac{1}{n}\sum_{i=1}^n g_i(x,y),
\end{aligned}
$}    
\end{equation}

where $f_i(x, y) = \mathbb{E}[f_i(x, y;\zeta^i)],\, g_i(x, y) = \mathbb{E}[g_i(x, y;\xi^i)]$ are stochastic upper- and lower-level loss functions of client $i$, and $n$ is the total number of clients. 

\begin{figure}
\centering
  \includegraphics[width=0.45\textwidth]{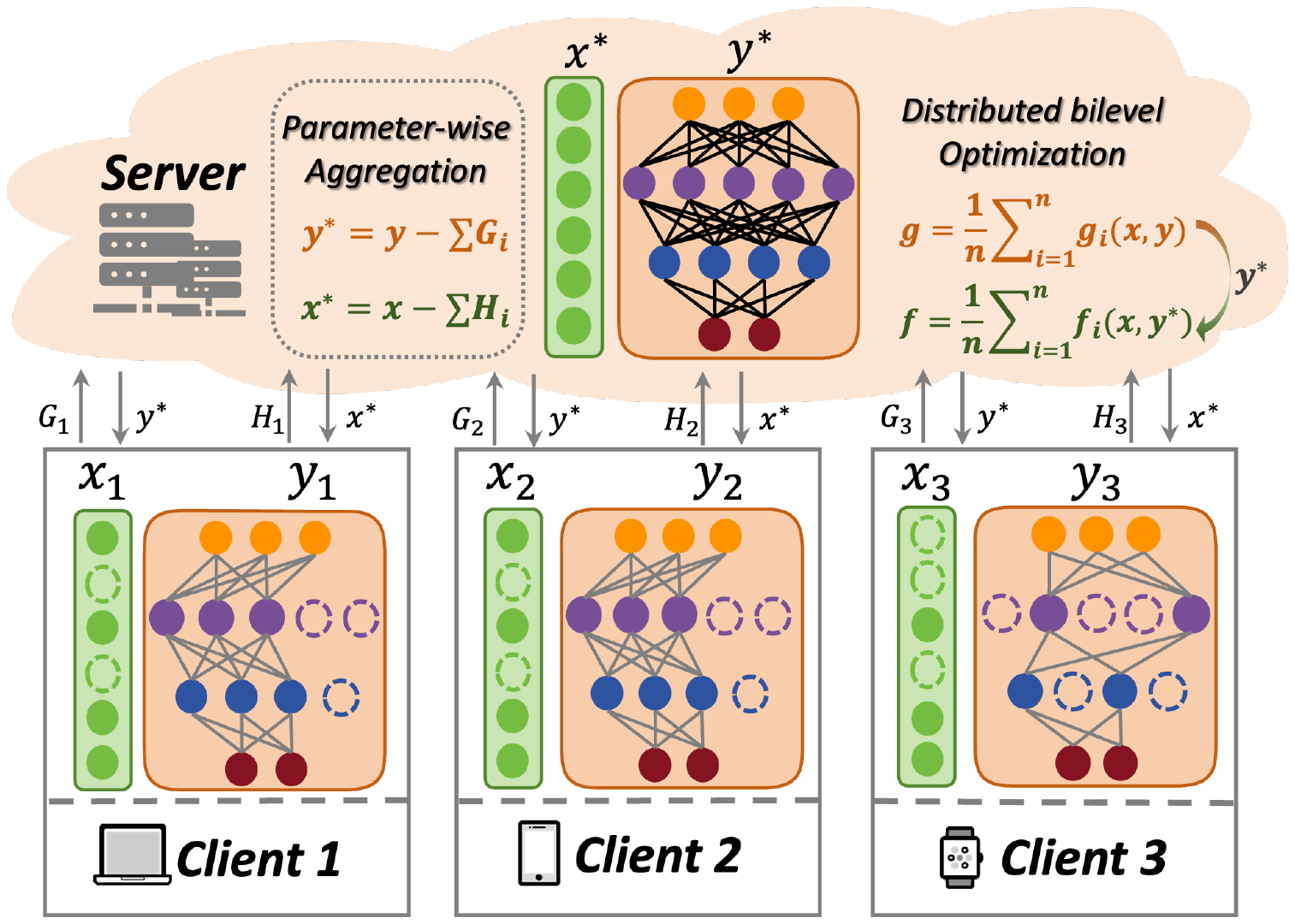}
  \caption{Illustration diagram of proposed resource-adaptive distributed bilevel optimization framework \textit{RABO}.}
  \label{fig:rabo}
\end{figure}

However, in real-world scenarios, with the rise of large-scale models~\cite{he2016deep,wang2024explicit,wen2023byzantine}, both the dimensions of outer parameters $x$ and inner parameters $y$ might become enormous, making it difficult to calculate not only the inner gradient $\nabla_y g_{i}(x,y)$, but also the hypergradient on resource-limited clients. Even the iterative differentiation (ITD) method~\cite{shaban2019truncated,domke2012generic} and the approximation implicit differentiation (AID) method~\cite{xiao2023communication,yang2024simfbo} 
have been widely used in both centralized and distributed scenarios, both of them still include the computation of higher-order derivatives and the matrix-vector multiplication, making the hypergradients estimation infeasible 
when local computation resource are limited. Therefore the main focus of our work arises: 
\textit{How to optimize the distributed bilevel problem with large-scale inner- and outer-models under resource-limited scenarios?} 
\IEEEpubidadjcol

We consider the following resource-aware bilevel optimization problem,
\begin{equation}
\scalebox{0.9}{$
\begin{aligned}\label{objective2}
&\min_{x\in\mathbb{R}^{d_1}} \Phi(x)=f(x, y^*(x)): = 
\frac{1}{n} \sum_{i=1}^n f_i(x\odot m_x^i, y^*(x)\odot m_y^i) 
\\& \mbox{s.t.} ~y^*(x)\in \arg \min_{y\in\mathbb{R}^{d_2}} g(x,y):=
\frac{1}{n}\sum_{i=1}^n g_i(x\odot m_x^i,y\odot m_y^i),
\end{aligned}
$}    
\end{equation}
where $m_x^i$ and $m_y^i$ denote the local masks of parameter $x$ and $y$ separately for client $i$, which can be dynamically adjusted based on available resources, $x\odot m_x^i$ and $y\odot m_y^i$ denote inner- and outer- submodel separately for client $i$, $f_i(\cdot, \cdot) = \mathbb{E}[f_i(\cdot, \cdot;\zeta^i)],\, g_i(\cdot, \cdot) = \mathbb{E}[g_i(\cdot, \cdot;\xi^i)]$ are stochastic upper- and lower-level loss functions of client $i$, and $n$ is the total number of clients. 
Nevertheless, when solving this bilevel optimization, some tricky challenges arise:
\textbf{(1) Complex coupling between inner-outer submodel training.} Training submodels introduces inherent randomness and variability in model parameters used in the both inner and outer optimization. This stochasticity creates a non-deterministic dependency between inner solutions and outer variables, complicating gradient estimation and bilevel coordination. 
\textbf{(2) Inaccurate surrogate approximation.} In bilevel optimization, the global surrogate variable $y$ is typically used to approximate the unavailable true inner optimal solution $y^*$. Submodel training introduces additional approximation errors and biases, further complicating the surrogate’s fidelity.
\textbf{(3) Intractable hypergradient computation.} The ITD and AID methods still include the computation of higher-order derivatives and the matrix-vector multiplication. It's challenging due to the submodel training introducing stochasticity and approximation errors that amplify numerical instability in higher-order derivative computations. And repeated Hessian-vector or Jacobian-vector products are still computationally intensive.

To fill these gaps, we propose a new resource-adaptive distributed bilevel optimization framework called \textit{RABO}.
Instead of training the full model, each client can adaptively train the submodels based on their local resources and reduce the total computations in optimization. 
The process in each global iteration is shown in Figure~\ref{fig:rabo}. We use the double-loop communication framework, each client first optimizes partial inner parameters $y_i$ and uploads the parameter variation, the server aggregates and distributes the global parameters $y^*$, then for outer parameters $x$. 
In order to further reduce the computation cost of hypergradient, we propose a partial hypergradient approximation method called \textit{RAFBO}. We use the idea of finite differences to approximate the partial hypergradient. It's second-order free, which avoids the computation of the Jaccobian and Hessian matrices. 
Otherwise, in theoretical analysis, we rigorously handle the coupled influence of partially trained $x$ and $y$ on the global hypergradient in the theoretical analysis. 
And we also reformulate the error bound between the global surrogate $y$ and the unavailable $y^*$ also needs to be reformulated due to the local partial training.
Thus, we establish an asymptotically optimal convergence rate $\mathcal{O}(1/\sqrt{\mathcal{C}^*_xQ})$ of proposed algorithms, where $Q$ is the global iterations and $\mathcal{C}^*_x$ is the minimum covering number of outer parameters defined in Section~\ref{sec:theory}. The extensive experimental results confirm the effectiveness of our proposed \textit{RABO} and \textit{RAFBO}. 

Our contributions can be summarized as follows:
\begin{itemize}
\item To the best of our knowledge, we are the first to propose the resource-adaptive distributed bilevel optimization framework. Each client can choose the appropriate inner and outer submodel adapted to their resource, then solve the global bilevel optimization collaboratively. 
\item To further alleviate the computational burden, we propose a second-order free hypergradient estimator, that uses finite differences to estimate partial hypergradient. It is possible to reduce the computation while keeping the estimation accuracy.
\item We provide a convergence analysis for the proposed \textit{RABO} and \textit{RAFBO} and carefully solve the coupled influence of partial outer parameters $x$ and inner parameters $y$. We prove that both of them can achieve an asymptotically optimal convergence rate $\mathcal{O}(1/\sqrt{\mathcal{C}^*_xQ})$ under the non-convex-strongly-convex condition. Existing studies would become a special case of our learning paradigm.
\item We conduct extensive experiments on two practical bilevel optimization tasks, by comparing with state-of-the-art single-level submodel training and distributed bilevel optimization algorithms. Results show our proposed methods achieve up to  $5\times$  and $4\times$ reduction in computation cost and communication cost compared to FedMBO. 
\end{itemize}




\section{Related Work}
Distributed bilevel optimization has emerged as a highly promising paradigm in collaboratively solving bilevel optimization. FedNest~\cite{conf/icml/fednest} is the first to focus on general nested problems and provides a convergence rate guarantee under the heterogeneity setting. And based on it, FedMBO~\cite{conf/icml/fedmbo} achieves the first linear speedup result. Considering the communication cost, AggITD~\cite{xiao2023communication} leverages the idea of iterative differentiation that computes the hypergradient through intermediate iterates of the lower-level updates on $y$, which achieves the same sample complexity as existing AID-based approaches while reducing the communication rounds. Further, SimFBO~\cite{yang2024simfbo} introduces an additional global quadratic programming and updates three sub-problem in clients simultaneously, thus per iteration only needs one communication round. Although ShroFBO~\cite{yang2024simfbo} considers the system-level heterogeneity, it claims that different computing capabilities lead to an unequal number of
local updates, which is different from the challenge that we consider. FedBiOAcc~\cite{li2024communication} leverages momentum-based variance reduction to achieve optimal convergence rate. Moreover, bilevel optimization has also been studied in other distributed settings, such as decentralized bilevel optimization~\cite{chen2023decentralized} and asynchronous bilevel optimization~\cite{jiaoasynchronous}.
However, existing distributed bilevel optimization algorithms require sufficient computation resources on clients, which cannot be applied in resource-limited scenarios.

\begin{table}[t]
 \caption{\textcolor{\attcolor}{Frequently used notations and corresponding descriptions}}
 \label{table:notation}
 \centering
	\begin{tabular}[t]{p{1.3cm} p{6.2cm}}
		\hline
		\textcolor{\attcolor}{Notations} & \textcolor{\attcolor}{Descriptions} \\ 
		\hline
        $x$ & outer parameter\\
        $y$ & inner parameter\\
        $n$ & total number of clients\\
        $f_i(\cdot,\cdot)$ & stochastic upper-level loss functions of client $i$\\
        $g_i(\cdot, \cdot)$&stochastic lower-level loss functions of client $i$\\
        $\Phi(x)$& Overall loss functions respect to $x$, $\Phi(x)=f(x, y^*(x)) $\\
        $m_{q,x}^i$& the mask of outer parameter in client $i$ in global iteration $q$ \\
        $m_{q,y}^i$& the mask of inner parameter in client $i$ in global iteration $q$ \\
        $R_i$ & resource constrains of client $i$\\
        $G$ & Inner gradient\\
        $\mathcal{H}$ & Hypergradient\\
        $\nabla g(x,y)$&The derivative of $g(x,y)$ with respect to y\\
        $\nabla_x f(x,y)$&The derivative of $f(x,y)$ with respect to $x$\\
        $\nabla_x f(x,y)$&The derivative of $f(x,y)$ with respect to $y$\\
        $\nabla_{x,y}^2 g(x,y)$&Hessian matrix of the $g(x,y)$ with respect to $x$ and $y$\\
        $\nabla_{y,y}^2 g(x,y)$&Hessian matrix of $g(x, y)$ with respect to $y$\\
		$\|\cdot\|$ & the vector $\ell_2$ norm or the matrix spectral norm depending on the argument\\
        $\mathcal{S}_x$ & the set of all outer parameters $x$\\
        $\mathcal{S}_y$ & the set of all inner parameters $y$\\
        $\mathcal{K}_{q,x}$ & the trained outer parameters set in global iteration $q$\\
        $\mathcal{K}_{q,y}$ & the trained inner parameters set in global iteration $q$\\
        $|\mathcal{K}_{q,x}|$ & the number of trained outer parameters in global iteration $q$\\
        $|\mathcal{K}_{q,y}|$ & the number of trained inner parameters in global iteration $q$\\
        $N_{q,x}^s$ & the set of clients training outer parameter $s$ in global iteration $q$\\
        $N_{q,y}^k$ & the set of clients training inner parameter $k$ in global iteration $q$\\
        $\mathcal{C}_{q,x}^s$ & $\mathcal{C}_{q,x}^s=|N_{q,x}^s|$ the number of clients in $N_{q,x}^s$\\
        $\mathcal{C}_{q,y}^k$ & $\mathcal{C}_{q,y}^k=|N_{q,y}^k|$ the number of clients in $N_{q,y}^k$\\
        $\mathcal{C}^*_x$   & minimum covering clients: $\mathcal{C}^*_x=\min\limits_{q,s}\mathcal{C}_{q,x}^s$, $s\in \mathcal{K}_{q,x}$, $\forall q$ \\
        $\mathcal{C}^*_y$   & minimum covering clients: $\mathcal{C}^*_y=\min\limits_{q,k}\mathcal{C}_{q,y}^k$, $k\in \mathcal{K}_{q,y}$, $\forall q$ \\
        $\alpha$& the learning rate of outer parameter\\
        $\beta$& the learning rate of inner parameter\\
        $\mathcal{P}$&The coordinate set\\
        $\mathbb{e}_p$&A coordinate vector\\
		\hline
	\end{tabular}
\end{table}

\begin{algorithm}[t]
  \caption{Recourse-Adaptive Distributed Stochastic Bilevel Optimization(RABO)}
  \label{alg:rabo} 
  \begin{algorithmic}[1]
    \STATE {\bfseries Input:} number of clients $n$, initial point $(x_0, y_0)$, mask policy ${M}(\cdot, R_n)$
    \FOR{$q = 0,1, \ldots, Q-1$}

    \FOR{$i \in \left[n\right]$ in parallel}
    \STATE \begin{small}Generate mask $ m^i_{q,x}=\mathcal{M}(x_{q},R_i)$ , $m^i_{q,y}=\mathcal{M}(y_{q}, R_i)$
    \STATE Generate submodel $ x^i_{q}=x_{q}\odot m^i_{q,x}, y^i_{q,0}=y_{q}\odot m^i_{q,y}$\end{small}

    \FOR{epoch $t = 0,1, \ldots, T-1$}
    \STATE Compute $G_{q,t}^i=\nabla g_i(x^i_{q}, y^i_{q,t};\xi_{t}^{i})\odot m^i_{q,y}$
    \STATE $y^i_{q,t+1}=y^i_{q,t}-\beta G_{q,t}^i$
    \ENDFOR
    
    \STATE $G_{q}^i=\frac{y^i_{q,T}-y^i_{q,0}}{\beta}$
    \ENDFOR
    \FOR{parameters $k$ to $\mathcal{S}_{y}$ in $y$}
    \STATE Find $N_{q,y}^{k}=\{i:m_{q,y}^k=1\}$
    \STATE $\mathcal{C}_{q,y}^{k}=|N_{q,y}^{k}|$
    \IF{$\mathcal{C}_{q,y}^{k}=0$}
    \STATE Update $y_{q+1}^k= y_{q}^k$
    \ELSE
    \STATE $y_{q+1}^k = y_{q}^k-\beta\frac{1}{\mathcal{C}_{q,y}^{k}} \sum_{n\in N_{q,y}^{k}}G_{q}^{i,k}$
    \ENDIF
    \ENDFOR
    \FOR{$i \in \left[n\right]$ in parallel}
    \STATE Generate submodel $y^i_{q+1}=y_{q+1}\odot m^i_{q,y}$
    \STATE \colorbox{cyan!30}{Compute hypergradient $\mathcal{H}_q^i=\nabla f_i(x_q^i, y_{q+1}^i;\zeta)$}
    
    \ENDFOR
    \FOR{parameters $s$ to $\mathcal{S}_{x}$ in $x$}
    \STATE Find $N_{q,x}^{s}=\{i:m_{q,x}^s=1\}$
    \STATE $\mathcal{C}_{q,x}^{s}=|N_{q,x}^{s}|$
    \IF{$\mathcal{C}_{q,x}^{s}=0$}
    \STATE Update $x_{q+1}^s= x_{q}^s$
    \ELSE
    \STATE $x_{q+1}^s = x_q^s- \alpha \frac{1}{\mathcal{C}_{q,x}^{s}} \sum_{n\in N_{q,x}^{s}} \mathcal{H}_q^i$
    \ENDIF
    \ENDFOR
    \ENDFOR
  \end{algorithmic}
\end{algorithm}

\section{Algorithm Design}
In this section, in order to reduce the training burden on resource-limited clients and make it possible to solve bilevel optimization problem, we propose a \textbf{R}esource-\textbf{A}daptive distributed \textbf{B}ilevel \textbf{O}ptimization framework called \textit{\textbf{RABO}}, which can adaptively train both the inner and outer parameters according to the clients' resources and gain the global model. Additionally, we introduce a new partial hypergradient efficient approximation method to further reduce the computation and memory overhead. The detailed descriptions are shown as follows.

\subsection{RABO}
The \textbf{Recourse-Adaptive Distributed Stochastic Bilevel Optimization (RABO)} algorithm is designed to solve bilevel optimization problems in a distributed setting, particularly for models with client heterogeneity and resource constraints. As shown in Algorithm~\ref{alg:rabo}, the server sends the latest full model to all clients, each client generates an adaptive submodel based on their resources. After the inner and outer partial parameters nested local updates, the server aggregates updates for each parameter. The above process will repeat until the model converges. The detailed algorithm proceeds in the $q$-th global iteration as follows:

\textbf{Local Submodel Generation.} 
At each global iteration \( q \), the server sends the latest global model $(x_q, y_q)$ to clients, each client uses the resource-adaptive mask policy $\mathcal{M}(\cdot, R_i)$ to generate masks, where $R_i$ is the resource constrains of the client $i$.
\begin{equation}
    m^i_{q,x}=\mathcal{M}(x_{q},R_i),\quad m^i_{q,y}=\mathcal{M}(y_{q}, R_i).
\end{equation}
The mask $m^i_{q,x}\in\{0,1\}^{d_1}$,  $m^i_{q,y}\in\{0,1\}^{d_2}$, each element denotes whether the parameter is in the local submodel or not. 
According to the mask, we use a structured pruning method to construct the corresponding submodel for further training.
\begin{equation}
    x^i_{q} = x_q \odot m^i_{q,x},\quad  y^i_{q,0} = y_q \odot m^i_{q,y},
\end{equation}

where $\odot$ operation denotes the structured pruning. Considering the heterogeneous dynamics of client resources, masks can be changed in different global iterations. Overall, we define the full inner parameters as $\mathcal{S}_{y}$, full outer parameters as $\mathcal{S}_{x}$ and trained ones as $\mathcal{K}_{q,y}, \mathcal{K}_{q,x}$.
Here, considering the actual application, we can determine whether to apply the mask for both inner and outer parameters or not. For example, when applying \textit{RABO} to hyper-parameter optimization, each client must optimize the full hyper-parameters $x$, thus, we can only prune inner parameters to match the local resource.

\textbf{Inner Submodel update.}
Each client calculates the inner gradient  $G_{q,t}^i$ based on the inner function $g_i(x^i_{q}, y^i_{q,t})$ and updates the inner submodel with $T$ epochs: 
\begin{equation}
    G_{q,t}^i = \nabla g_i(x^i_{q}, y^i_{q,t}; \xi_t^i) \odot m^i_{q,y},\quad y^i_{q,t+1} = y^i_{q,t} - \beta G_{q,t}^i   .
\end{equation}
After completing the local updates, compute the gradient \( G_q^i \) for each client \( i \) as:
$G_q^i = (y^i_{q,T} - y^i_{q,0})/\beta$.
This gradient is used for the global inner optimization step.

\textbf{Inner Parameter-wise Gradient Aggregation}
For each parameter $k \in \mathcal{S}_y$ in $y$, identify the set of clients $N_{q,y}^k=\{ i:m_{q,y}^k=1\}$ that contribute to the update of this parameter, the number of these clients denotes as $\mathcal{C}_{q,y}^k$. The aggregated gradient is then used to update $y_{q+1}^k$ as:
\begin{equation}
y_{q+1}^k = y_q^k - \beta \frac{1}{\mathcal{C}_{q,y}^k} \sum_{i \in N_{q,y}^k} G_q^{i,k},
\end{equation}
where $\beta$ is the inner learning rate. If no clients contribute to a particular parameter, it remains unchanged. Thus, we gain the updated inner parameter $y_{q+1}$ to approximate the minimizer $y^*(x)$ of lower-level problem.

\textbf{Hypergradient Calculation.} 
The server broadcasts the updated $y_{q+1}$, and the clients use it to calculate the hypergradient.
\begin{align}
\nabla f_i(x_q^i, y_{q+1}^i) = &\nabla_x f_i(x_q^i, y_{q+1}^i)- \nabla^2_{xy} g_i(x_q^i, y_{q+1}^i)\\
&[\nabla^2_{yy} g_i(x_q^i, y_{q+1}^i)]^{-1}\nabla_y f_i(x_q^i, y_{q+1}^i).\nonumber\label{surrogate_f}
\end{align}
Due to the submodels being trained within each client, 
the dimensions of parameters are much smaller than the full model, which reduces the computational burden. 

\textbf{Outer Parameter-wise Gradient Aggregation.} We use the same parameter-wise gradient aggregation as inner update to compute the aggregated hypergradients and update the outer parameters. 
\begin{equation}
 x_{q+1}^s = x_q^s - \alpha \frac{1}{\mathcal{C}_{q,x}^s} \sum_{i \in N_{q,x}^s} \nabla f_i(x_q^i, y_{q+1}^i; \zeta), 
\end{equation}
where \( \alpha \) is the outer learning rate, and \( \mathcal{C}_{q,x}^s \) denotes the number of clients contributing to the update.

\begin{algorithm}[t]
  \caption{\colorbox{cyan!30}{Second-Free Hypergradient Approximator}}
  \label{alg:approximation} 
  \begin{algorithmic}[1]
    \STATE {\bfseries Input:} submodel $x_q^i, y_{q+1}^i$, perturbation vector $\mu$
    \STATE {\bfseries Output:} Approximated local hypergradient $\mathcal{H}_i$
    \STATE Calculate and sample coordinate set $\mathcal{P}_{q}^i$ of trained submodel $x_q^i$.
    \FOR{$\textbf{e}_p\in\mathcal{P}_{q}^i$}
    \STATE $\delta(x_q,\textbf{e}_p, \xi)=\frac{\nabla_y g(x_q^i+\mu\textbf{e}_p,y_{q+1}^i, \xi)-\nabla_y g(x_q^i,y_{q+1}^i, \xi)}{\mu}$
    \ENDFOR
    \STATE $\mathcal{H}_i=\hat{\nabla} f_i(x_q,y_{q+1}, \zeta)=\nabla_x f_i(x_q^i,y^i_{q+1}, \zeta)+\sum_{\mathcal{P}_{q}^i}\left \langle\delta(x_q,\textbf{e}_p, \xi),\nabla_yf_i(x_q^i,y^i_{q+1}, \zeta)\right \rangle \textbf{e}_p $
  \end{algorithmic}
\end{algorithm}
\subsection{RAFBO}\label{sec:rafbo}
In this section, we give a further partial approximation method for the local hypergradient to avoid
computing the Jacobian and Hessian matrix, as described in Algorithm \ref{alg:approximation}. Since this method is second-order free, we call the algorithm with the hypergradient estimator as \textit{RAFBO}.
Returning to the explicit function of hypergradient,
\begin{align}
    \nabla\Phi(x) = \nabla_x f(x,y^*(x))+\mathcal{J}^*(x)^\top\nabla_y f(x,y^*(x)),
\end{align}
where $\mathcal{J}^*(x)=\frac{\partial y*(x)}{\partial x}$. Following ~\cite{pzobo}, we use the idea of finite differences to approximate the $\mathcal{J}^*(x)$ directly, which avoids calculating the $\nabla^2_{xy} g_i(x_q^i, y_{q+1}^i)$ and $[\nabla^2_{yy} g_i(x_q^i, y_{q+1}^i)]^{-1}$. Inspired by the ~\cite{lian2016comprehensive, liu2020primer}, we choose the deterministic coordinate-wise perturbations. First we calculate and sample the coordinate set \( \mathcal{P}_q^i \) for the trained submodel \( x_q^i \), which consists of the directions along which the perturbations will be applied. Compared to the existing finite differences method, we use the deterministic coordinate-wise perturbations of model instead of random perturbations, which is more accurate and time efficient~\cite{chen2024deepzero}.
It is worth noting that due to the construction of submodel, the size of coordinate set is smaller than the original full model, which further reduces the computation. 

For each coordinate \( \textbf{e}_p \in \mathcal{P}_q^i \), we apply a perturbation \( \mu \textbf{e}_p \) to the submodel \( x_q^i \) and compute the difference in the gradient of the loss function with respect to \( y \):
\begin{equation}
\scalebox{0.9}{$
\begin{aligned}
\hat{\mathcal{J}}(x_q, \textbf{e}_p, \xi) = \frac{\nabla_y g(x_q^i + \mu \textbf{e}_p, y_{q+1}^i, \xi) - \nabla_y g(x_q^i, y_{q+1}^i, \xi)}{\mu} \textbf{e}_p,
\end{aligned}
$}    
\end{equation} 
where $\mu>0$ is the smoothing parameter. Once the finite differences are computed for all coordinates in \( \mathcal{P}_q^i \), the local hypergradient is approximated as:
\begin{equation}
\scalebox{0.9}{$
\begin{aligned}
  &\hat{\nabla} f_i(x_q,y_{q+1}, \zeta)\\
  &=\nabla_x f_i(x_q^i,y^i_{q+1}, \zeta)+\sum_{\mathcal{P}_{q}^i} \hat{\mathcal{J}}(x_q,\textbf{e}_p, \xi)^\top\nabla_yf_i(x_q^i,y^i_{q+1}, \zeta)\nonumber\\
  &=\nabla_x f_i(x_q^i,y^i_{q+1}, \zeta)+\sum_{\mathcal{P}_{q}^i}\left \langle\delta(x_q,\textbf{e}_p, \xi),\nabla_yf_i(x_q^i,y^i_{q+1}, \zeta)\right \rangle \textbf{e}_p \nonumber 
\end{aligned}
$}    
\end{equation}  
From the perspective of computational efficiency, we use the vector-vector dot product instead of Jacobian-vector product, where $\delta(x_q,\textbf{e}_p, \xi)=(\nabla_y g(x_q^i+\mu\textbf{e}_p,y_{q+1}^i, \xi)-\nabla_y g(x_q^i,y_{q+1}^i, \xi))/\mu\in \mathbb{R}^{d_2}$, $\left\langle \cdot, \cdot \right\rangle$ denotes the inner product. 
By embedding this approximation to the proposed \textit{RABO}, we reduce the computational complexity of calculating the hypergradient while maintaining a reasonable level of accuracy for the parameter updates in the outer loop.

\section{Theoretical Analysis}\label{sec:theory}
In this section, we show the convergence rate of our proposed \textit{RABO} and \textit{RAFBO}. Firstly, we provide some standard definitions and mild assumptions needed for our analysis.

\subsection{Assumptions and Definitions}
Let $z=(x,y)\in \mathbb{R}^{d_1+d_2}$ denote all parameters throughout the following definitions and assumptions.

\begin{definition}
    Function $h$ is $L$-Lipschitz continuous if
    \begin{equation}
    \|h(z_1)-h(z_2)\|\leq L\|z_1-z_2\|\quad \forall z_1, z_2.    
    \end{equation}
\end{definition}

\begin{definition}
    The solution $z$ is $\epsilon$-accurate stationary point of the objective function $\Phi(z)$ if $\mathbb{E}\|\nabla\Phi(z)\|^2\leq \epsilon$, where $z$ is the output of an algorithm.
\end{definition}

\begin{definition}
(Minimum covering clients). The minimum clients trained the i-th parameter across all submodels of inner parameters $y$ or outer parameters $x$ throughout all global iterations can be defined as:
\begin{align}
&\mathcal{C}^*_x= \min\limits_{q,s} \mathcal{C}_{q,x}^s, s \in \mathcal{K}_{q,x}, \forall q,\\
&\mathcal{C}^*_y= \min\limits_{q,k} \mathcal{C}_{q,y}^k, k \in \mathcal{K}_{q,y}, \forall q,
\end{align}
where $\mathcal{C}_{q,x}^s$ is the number of the client that train the $i$-th parameter of outer parameters $x$ in the global iteration $q$ and $\mathcal{C}_{q,y}^k$ is related to inner parameters $y$.
\end{definition}
\begin{assumption}(Non-convex-strongly-convex).
\label{assum:convex}
For all $i\in [n]$, $f_i(x,y)$ is possibly non-convex and $g_i(x,y)$ is $\mu_g$-strongly convex in $y$ for any fixed $x\in\mathbb{R}^{d_1}$.
\end{assumption}

\begin{assumption}(Lipshitz continuous).
\label{assum:lipschitz}
For all $i\in[n]$, $f_i(z)$, $\nabla f_i(z)$, $\nabla g_i(z)$, $\nabla^2 g_i(z)$ are $\ell_{f,0}$, $\ell_{f,1}$, $\ell_{g,1}$, $\ell_{g,2}$-Lipshitz continuous, respectively.
\end{assumption}

\begin{assumption}\label{assum:sto_sample}(Unbiased estimators).
For all $i\in [n]:$
$\nabla f_i(z;\zeta)$, $\nabla g_i(z;\xi)$ are unbiased estimators of $\nabla f_i(z)$, $\nabla g_i(z)$, there exist constants $\sigma_f^2$, $\sigma^2_{g}$, such that 
\begin{align}
    &\mathbb{E}_{\zeta}\left\|\nabla f_i(z;\zeta)-\nabla f_i(z)\right\|^2\leq \sigma_f^2,
    \\
    &\mathbb{E}_{\xi}\left\|\nabla g_i(z;\xi)-\nabla g_i(z)\right\|^2\leq \sigma_{g}^2,
\end{align}
\end{assumption}

\begin{assumption}\label{assum:hetero}
     (Bounded data heterogeneity level). For all $i\in[n]$: there exist constants $\delta^2_f$ and $\delta^2_g$, such that:
     \begin{align}
         \mathbb{E}\left\|\nabla f_i(z)-\nabla f(z)\right\|^2]\leq \delta^2_f,\\
         \mathbb{E}\left\|\nabla g_i(z)-\nabla g(z)\right\|^2]\leq \delta^2_g.        
     \end{align}
\end{assumption}

\begin{assumption}\label{assum:mask}
    (Bounded noise induced from submodel). The deviation of the submodel in client $n$ from the original parameters for every global iteration $q$ is limited by $w_1,w_2\in[0,1)$:
\begin{align}
        \|x_q-x_q\odot m_{q,x}^i\|^2 &\leq w_1^2\|x_q\|^2,\\
        \|y_q-y_q\odot m_{q,y}^i\|^2 &\leq w_2^2\|y_q\|^2.
\end{align}
\end{assumption}

The above assumptions are commonly used in previous works. Assumption~\ref{assum:convex} supposes the non-convex-strongly-convex bilevel optimization problems, which guarantee the existence of the hypergradient $\nabla\Phi(x)$~\cite{kwon2023fully,FdeHBO}. 
Assumption~\ref{assum:lipschitz} imposes the levels of Lipschitz smoothness, which is a standard assumption in distributed bilevel optimization~\cite{xiao2023communication,yang2024simfbo}.
Assumption~\ref{assum:sto_sample} shows the stochastic estimators are unbiased with bounded variances, which is widely used in stochastic optimization~\cite{tarzanagh2022fednest,conf/icml/fedmbo}. 
Assumption~\ref{assum:hetero} controls the difference between local and global gradient, which is common in traditional federated learning without bilevel optimization~\cite{li2019convergence,yuan2022convergence}.
Assumption~\ref{assum:mask} states the influence of the submodel, as also adopted by other studies about model pruning~\cite{wang2023theoretical,zhou2024every}.

In order to analyze the convergence rate of our proposed recourse-adaptive distributed stochastic bilevel optimization algorithm, we firstly state some preliminary lemmas as follows:

\begin{lemma}
Under Assumptions 1 and 2, we have 
\begin{align}
    ||\nabla \Phi(x_1) -\nabla \Phi(x_2)||&\leq L_f ||x_1 - x_2||,\label{eq:f_lip}\\
    ||y^*(x_1)-y^*(x_2)||&\leq L_y ||x_1-x_2||,
\end{align}
where 
\begin{align}
    L_f :=& l_{f,1}+\frac{l_{g,1}(l_{f,1}+M_f)}{\mu_g}+\frac{l_{f,0}}{\mu_g}(l_{g,2}+\frac{l_{g,1}l_{g,2}}{\mu_g}), \nonumber
    \\
    L_y :=& \frac{l_{g,1}}{\mu_g}.\nonumber
\end{align}
For all $i\in [n]$, we have 
\begin{align}
    ||\nabla f_i(x_1, y)-\nabla f_i(x_1,y^*(x_1))|| &\leq M_f ||y-y^*(x_1)||,\\
    ||\nabla f_i(x_1,y)-\nabla f_i(x_2,y)|| &\leq M_f||x_1-x_2||,
\end{align}
where the constant $M_f$ is given by 
\begin{align}
    M_f:=l_{f,1}+\frac{l_{g,1}l_{f,1}}{\mu_g}+\frac{l_{f,0}}{\mu_g}(l_{g,2}+\frac{l_{g,1}l_{g,2}}{\mu_g})\nonumber
\end{align}
\end{lemma}

\begin{lemma} Under Assumptions 1 and 2, we have 
\begin{align}\mathbb{E}\|\nabla f^s_n(x_{q}^n,y_{q+1}^n)-\nabla f^s_n(x_{q},y_{q+1})\|^2\leq \mathcal{E}_h\end{align}
\end{lemma}
\begin{proof}
\begin{align*}
    &\mathbb{E}\|\nabla f^s_n(x_{q}^n,y_{q+1}^n)-\nabla f^s_n(x_{q},y_{q+1})\|^2\\
    \leq&2\mathbb{E}\|\nabla f^s_n(x_{q}^n,y_{q+1}^n)-\nabla f^s_n(x_{q},y_{q+1}^n)\|^2\\&\quad+2\mathbb{E}\|\nabla f^s_n(x_{q},y_{q+1}^n)-\nabla f^s_n(x_{q},y_{q+1})\|^2\\
    \leq &2 M_f^2\|x_q^n-x_q\|^2+2M_f^2\|y_{q+1}^n-y_{q+1}\|^2 \\
    \leq &2M_f^2 w_1^2 \|x_q\|^2+2M_f^2 w_2^2 \|y_{q+1}\|^2
\end{align*}
\end{proof} 

\begin{lemma}\label{hyper-gradient estimation}
(Hypergradient estimation error.) Suppose that Assumption~\ref{assum:lipschitz} holds. The hypergradient estimation error can be bounded as 
\begin{align}
    \|\hat{\nabla} f_i(x_q,y_{q+1},\zeta)- \nabla f_i(x_q,y_{q+1},\zeta)\|^2\leq \sigma_{\mathcal{J}}^2.
\end{align}
where $\sigma_{\mathcal{J}}^2 = P^{*2}l_{g,1}^2\mu^2l_{f,0}^2/4$, $P^*=\max\limits_{i}|\mathcal{P}_q^i|,i\in[n], \forall q.$
\end{lemma}
\subsection{Convergence Result of RABO and RAFBO}
\begin{theorem}\label{convergence1}
(Convergence result of RABO) Let all assumptions hold, suppose the learning rates satisfy these conditions
\begin{equation*}\label{}
   \left\{\begin{array}{l c l}
    \alpha\leq\frac{1}{L_f+4M_f}\\
    \beta\leq min\{\frac{1}{2l_{g,1}}, \frac{1}{\mu_g}\} \nonumber\\
     \beta\geq\frac{1}{\mu_g}-\frac{1}{2\alpha L_yM_f\mu_g}\nonumber
    \end{array}\right.
\end{equation*}
Then for all $Q\geq1$, we have
\begin{align}
     \frac{1}{Q}\sum_{q=1}^Q&\sum_{s\in \mathcal{K}_{q,x}}\mathbb{E}\|\nabla \Phi^s(x_{q})\|^2\\
     \leq &\frac{2\Delta}{\alpha Q}+\frac{2N\hat{\alpha}_1}{\alpha\mathcal{C}_x^*}\sigma_f^2+\frac{2N\hat{\alpha}_1}{\alpha\mathcal{C}_x^*}\delta^2_f\nonumber\\
     &+\frac{4N\hat{\alpha}_1}{\alpha\mathcal{C}_x^*}\frac{1}{Q}\sum_{q=1}^Q(M_f^2 w_1^2 \|x_q\|^2+M_f^2 w_2^2 \|y_{q+1}\|^2)\nonumber\\
     &+(4M_f^2+\frac{4M_f}{\alpha L_y})\frac{4\beta^2NT(\sigma_{g}^2+\delta_{g}^2)}{\mathcal{C}^{*}_{y}}\nonumber
\end{align}
where $\Delta=\mathbb{E}\|V_{Q+1}\|-\mathbb{E}\| V_{1}\|$, $\hat{\alpha}_1=2L_f\alpha^2+8M_f\alpha^2$
\end{theorem}
Theorem~\ref{convergence1} shows the convergence rate of the proposed algorithm \textit{RABO} by presenting the upper bound regarding the globally averaged hypergradient for all model parameters.
\begin{remark}
(Impact of different key factors). From the Theorem~\ref{convergence1}, we
analyze the impact of key factors on the \textit{RABO} convergence rate:
\begin{itemize}
    \item \textbf{Impact of noise induced from mask.} As the Theorem~\ref{convergence1} shows, 
    due to the submodel training strategy, we additionally introduce the term $\mathcal{O}(\frac{\alpha}{\mathcal{C}_x^*}\frac{1}{Q}\sum_{q=1}^QM_f^2 w_1^2 \|x_q\|^2+\frac{\alpha}{\mathcal{C}_x^*}\frac{1}{Q}\sum_{q=1}^QM_f^2 w_2^2 \|y_{q+1}\|^2)$ of the hypergradient bound, which is proportional to the noise $w_1$ and $w_2$. When $w_1$ and $w_2$ are small, the convergence rate is faster under the same $\mathcal{C}^*_x$.
    Notably, if the client trains the full parameter $x$ and $y$, where $\mathcal{K}_{q,x}=\mathcal{S}_x, \mathcal{K}_{q,y}=\mathcal{S}_y$ and $m_x^i=1, m_y^i=1$, then $w_1=0, w_2=0$, there is also no additional noise generated.
    \item \textbf{Impact of data heterogeneity $\delta_f$ and $\delta_g$.} As previously described, data distribution in real-world settings is inherently non-IID. This heterogeneity plays a critical role in determining the convergence behavior of the model, as shown in Theorem~\ref{convergence1}. Here we make a detailed distinction between the effects of data heterogeneity on inner and outer functions and denote as $\delta_f$ and $\delta_g$. 
    Specifically, larger $\delta_f$ and $\delta_g$ indicate a higher degree of data heterogeneity, which adversely affects convergence by slowing it down. In contrast, when the data distribution is IID, such that $\delta_f = 0, \delta_g=0$, the impact of heterogeneity vanishes.
\end{itemize}
\end{remark}

Next, by specifying the parameters $\alpha$ and $\beta$, we can obtain the following convergence result.
\begin{corollary}\label{corollary1}
(Convergence rate of RABO) Let all assumptions hold. Supposing  $\alpha\leq \sqrt{C_x^*/Q}$, and $\beta \leq \frac{1}{NTQ}$, we can get
\begin{align}
    &\frac{1}{Q}\sum_{q=1}^Q\sum_{s\in \mathcal{K}_{q,x}}\mathbb{E}\|\nabla \Phi^s(x_{q})\|^2\\&\leq \mathcal{O}(\frac{1}{\sqrt{\mathcal{C}_x^*Q}}+\frac{1}{NTQ^2\mathcal{C}^{*}_{y}}+\frac{1}{NTQ^{1.5}\sqrt{\mathcal{C}_x^*}\mathcal{C}^{*}_{y}})\nonumber
\end{align}
\end{corollary}
\begin{remark}
   Corollary~\ref{corollary1} shows that the proposed resource-adaptive distributed bilevel optimization algorithm \textit{RABO} achieves a convergence rate $\mathcal{O}(\frac{1}{\sqrt{\mathcal{C}_x^*Q}})$. It demonstrates that for any local inner update epoch $T$, once choosing the proper learning rate $\alpha$ and $\beta$, \textit{RABO} converges with a sub-linear rate and has a linear speedup with respect to the minimum covering clients of outer parameter $\mathcal{C}_x^*$, which denotes the minimal number of clients that train specific partial outer parameters.
   When $\mathcal{C}_x^*=N$, all outer parameters are trained by all clients in every global iteration, the convergence rate of \textit{RABO} achieves $\mathcal{O}(\frac{1}{\sqrt{NQ}})$, which is identical to the linear speedup convergence rate of distributed bilevel optimizaiton algorithm \textit{FedMBO}~\cite{conf/icml/fedmbo}. 

   Notably, the definition of $\nabla\Phi(x)$ is under the condition of $y=y^*(x)$. Thus the inner submodel training will not dominate the convergence. But from lemma~\ref{lemma:error_inner} and Corollary~\ref{corollary1}, we can find that training the partial inner parameter also influences the convergence regarding the minimum covering clients of inner parameter $\mathcal{C}_y^*$. As with $\mathcal{C}_x^*$, it means inner parameters can be trained by at least $\mathcal{C}_y^*$ clients in each global iteration. As $\mathcal{C}_y^*$ increases while fixing other settings, the more frequently the model can be trained, thus the model can converge faster.
\end{remark}

\begin{theorem}\label{convergence2}
(Convergence result of RAFBO) Let all assumptions hold. Supposing $\alpha\leq \sqrt{\mathcal{C}_x^*/Q}$, $\beta \leq \frac{1}{NTQ}$, we can get
\begin{align}
     \frac{1}{Q}\sum_{q=1}^Q&\sum_{s\in \mathcal{K}_{q,x}}\mathbb{E}\|\nabla \Phi^s(x_{q})\|^2\\
     \leq& \mathcal{O}(\frac{\Delta}{\sqrt{\mathcal{C}_x^*Q}}+\frac{\sigma_{\mathcal{J}}^2+\sigma_f^2+\delta^2_f}{\sqrt{\mathcal{C}_x^*Q}}+\frac{w_1A_1+w_2A_2}{\sqrt{\mathcal{C}_x^*Q}}\nonumber\\
     &+\frac{\sigma_{g}^2+\delta_{g}^2}{NTQ^{1.5}\sqrt{\mathcal{C}_x^*}\mathcal{C}^{*}_{y}})\nonumber
\end{align}
where $\Delta=\mathbb{E}\|V_{Q+1}\|-\mathbb{E}\| V_{1}\|$, $A_1=\frac{1}{Q}\sum_{q=1}^Q\|x_q\|^2$, $A_2=\frac{1}{Q}\sum_{q=1}^Q\|y_{q+1}\|^2$.
\end{theorem}
\begin{remark}
Compared to the original algorithm \textit{RABO}, the bound of averaged global hypergradient only has one additional term $\mathcal{O}(\frac{\sigma_{\mathcal{J}}^2}{\sqrt{\mathcal{C}_x^*Q}})$, which is introduced by the partial hypergradient approximation according to Lemma~\ref{hyper-gradient estimation}. It has the same order $\mathcal{O}(\frac{1}{\sqrt{\mathcal{C}_x^*Q}})$ as the dominant term, and the total convergence rate is identical to \textit{RABO}. The approximation is designed to reduce computational memory without sacrificing the effectiveness of the optimization process. The theoretical results further show that the convergence rate of full model is not significantly altered by the approximation.
\end{remark}

\subsection{Convergence Analysis of RABO}
In this part, we provide concrete proofs and analyzes for
the important lemma and theorem we have mentioned.
\begin{lemma}\label{lemma:error_inner}
(Inner error bound). Let all assumptions hold and $y^*(x_q)=\arg\min\limits_y \frac{1}{N}\sum_{i=1}^N g_i(x_q,y)$, $\beta\leq1/2l_{g,1}$, we have,
 \begin{align}
     \mathbb{E}&\|y_{q+1}-y^*(x_q)\|^2\\ &\leq(1-\beta\mu_g)\mathbb{E}\|y_{q}-y^*(x_q)\|^2+\frac{4NT\beta^2}{\mathcal{C}_y^*}(\sigma_{g}^2+\delta_{g}^2)\nonumber
    \end{align}
\end{lemma}
\begin{proof}
\begin{align}
    &\mathbb{E} \left\|y_{q+1}-y^*(x_q)\right\|^2
    \nonumber
    \\
    =&\sum_{k\in \mathcal{K}_{q,y}}\mathbb{E}\|y_{q+1}^k-y^*(x_q)^k\|^2 
    \nonumber 
    \\
    =& \sum_{k\in \mathcal{K}_{q,y}}\mathbb{E}\|y_{q}^k-\frac{1}{\mathcal{C}^{s}_{q,y}}\sum_{j\in N^{s}_{q,y}}\sum_{t=1}^T\beta \nabla_y g^k_j(x_{q}^j,y_{q,t}^j,\xi_t^j)-y^*(x_q)^k\|^2 
    \nonumber
    \\
    =&\sum_{k\in \mathcal{K}_{q,y}}\mathbb{E}\left\|y_{q}^k-y^*(x_q)^k\right\|^2\nonumber\\&-2\beta\sum_{k\in \mathcal{K}_{q,y}}\mathbb{E}\langle \frac{1}{\mathcal{C}^{s}_{q,y}}\sum_{j\in N^{s}_{q,y}}\sum_{t=1}^T \nabla_y g^k_j(x_{q}^j,y_{q,t}^j,\xi_t^j),\nonumber\\&\quad y_{q}^k-y^*(x_q)^k\rangle\nonumber\\
    &+\beta^2\sum_{k\in \mathcal{K}_{q,y}}\mathbb{E}\|\frac{1}{\mathcal{C}^{s}_{q,y}}\sum_{j\in N^{s}_{q,y}}\sum_{t=1}^T \nabla_y g^k_j(x_{q}^j,y_{q,t}^j,\xi_t^j)\|^2
    \nonumber
    \\
    \leq& (1-\beta\mu_g)\sum_{k\in \mathcal{K}_{q,y}}\left\|y_{q}^k-y^*(x_q)^k\right\|^2\label{lemma:bounded-gradient}\\&-2\beta\frac{NT}{\mathcal{C}^*} \sum_{k\in \mathcal{K}_{q,y}}\mathbb{E}\left[g(x_q^k,y_{q,t}^k)-g(x_q^k,y^*(x_q)^k)\right]\nonumber\\
    &+\beta^2\sum_{k\in \mathcal{K}_{q,y}}\mathbb{E}\|\frac{1}{\mathcal{C}^{s}_{q,y}}\sum_{j\in N^{s}_{q,y}}\sum_{t=1}^T \nabla_y g^k_j(x_{q}^j,y_{q,t}^j,\xi_t^j)\|^2\nonumber
\end{align}
To bound the last term in Eq.~\ref{lemma:bounded-gradient}, we have
\begin{align}
    &\sum_{k\in \mathcal{K}_{q,y}}\mathbb{E}\|\frac{1}{\mathcal{C}^{s}_{q,y}}\sum_{j\in N^{s}_{q,y}}\sum_{t=1}^T \nabla_y g^k_j(x_{q}^j,y_{q,t}^j,\xi_t^j)\|^2
    \nonumber
    \\
    =&\sum_{k\in \mathcal{K}_{q,y}}\mathbb{E}\|\frac{1}{\mathcal{C}^{s}_{q,y}}\sum_{j\in N^{s}_{q,y}}\sum_{t=1}^T[ \nabla_y g^k_j(x_{q}^j,y_{q,t}^j,\xi_t^j)\nonumber\\&-\nabla_y g^k_j(x_{q}^j,y^*(x_{q}),\xi_t^j)+\nabla_y g^k_j(x_{q}^j,y^*(x_{q}),\xi_t^j)]\|^2
    \nonumber
    \\
    \leq& \frac{2}{\mathcal{C}^{s}_{q,y}}\sum_{j\in N^{s}_{q,y}}\sum_{t=1}^T\sum_{k\in \mathcal{K}_{q,y}}\mathbb{E}\|\nabla_y g^k_j(x_{q}^j,y_{q,t}^j,\xi_t^j)\nonumber\\
    &-\nabla_y g^k_j(x_{q}^j,y^*(x_{q}),\xi_t^j)\|^2\nonumber\\
    &+2\sum_{k\in \mathcal{K}_{q,y}}\mathbb{E}\|\frac{1}{\mathcal{C}^{s}_{q,y}}\sum_{j\in N^{s}_{q,y}}\sum_{t=1}^T\left[\nabla_y g^k_j(x_{q}^j,y^*(x_{q}),\xi_t^j)\right]\|^2
    \nonumber
    \\
    \leq& \frac{4l_{g,1}}{\mathcal{C}^{s}_{q,y}}\sum_{j\in N^{s}_{q,y}}\sum_{t=1}^T\sum_{k\in \mathcal{K}_{q,y}}\mathbb{E}[g^k_j(x_{q}^j,y_{q,t}^j,\xi_t^j)-g^k_j(x_{q}^j,y^*(x_{q}),\xi_t^j)\nonumber\\
    &-\langle \nabla_y g^k_j(x_{q}^j,y^*(x_{q}),\xi_t^j),y_{q,t}^j-y^{*}(x^q)\rangle]
    \nonumber
    \\
    &+2\sum_{k\in \mathcal{K}_{q,y}}\mathbb{E}\|\frac{1}{\mathcal{C}^{s}_{q,y}}\sum_{j\in N^{s}_{q,y}}\sum_{t=1}^T[\nabla_y g^k_j(x_{q}^j,y^*(x_{q}),\xi_t^j)]\|^2
    \nonumber
    \\=&4l_{g,1}\frac{NT}{\mathcal{C}^{*}_{y}}\sum_{k\in \mathcal{K}_{q,y}}\mathbb{E}\left[g\left(x^k,y^{k,t}\right)-g\left(x^k,y^*(x^k)\right)\right]\label{lemma:gradient-y*}\\
    &+2\sum_{k\in \mathcal{K}_{q,y}}\mathbb{E}\|\frac{1}{\mathcal{C}^{s}_{q,y}}\sum_{j\in N^{s}_{q,y}}\sum_{t=1}^T[\nabla_y g^k_j(x_{q}^j,y^*(x_{q}),\xi_t^j)]\|^2\nonumber
\end{align}
To bound the last term in Eq.~\ref{lemma:gradient-y*}, we have
\begin{align}
    &\sum_{k\in \mathcal{K}_{q,y}}\mathbb{E}\|\frac{1}{\mathcal{C}^{s}_{q,y}}\sum_{j\in N^{s}_{q,y}}\sum_{t=1}^T\left[\nabla_y g^k_j(x_{q}^j,y^*(x_{q}),\xi_t^j)\right]\|^2
    \nonumber
    \\
    \leq&\frac{2}{\mathcal{C}^{s}_{q,y}}\sum_{j\in N^{s}_{q,y}}\sum_{t=1}^T\sum_{k\in \mathcal{K}_{q,y}}\mathbb{E}\|\nabla_y g^k_j(x_{q}^j,y^*(x_{q}),\xi_t^j)\nonumber\\&-\nabla_y g^k_j(x_{q}^j,y^*(x_{q}))\|^2
    \nonumber
    \\
    &+\frac{2}{\mathcal{C}^{s}_{q,y}}\sum_{j\in N^{s}_{q,y}}\sum_{t=1}^T\sum_{k\in \mathcal{K}_{q,y}}\mathbb{E}\|\nabla_y g^k_j(x_{q}^j,y^*(x_{q}))\nonumber\\&-\nabla_y g\left(x_{q},y^*(x_{q})\right)\|^2\nonumber\\
    \leq& \frac{2NT(\sigma_{g}^2+\delta_{g}^2)}{\mathcal{C}^{*}_{y}}
    \nonumber
\end{align}
Thus, we can gain that
\begin{align}
    &\mathbb{E} \left\|y_{q+1}-y^*(x_q)\right\|^2
    \nonumber
    \\
    \leq& (1-\beta\mu_g)\sum_{k\in \mathcal{K}_{q,y}}\left\|y_{q}^k-y^*(x_q)^k\right\|^2\nonumber\\&-2\beta\frac{NT}{\mathcal{C}^*} \sum_{k\in \mathcal{K}_{q,y}}\mathbb{E}\left[g(x_q^k,y_{q,t}^k)-g(x_q^k,y^*(x_q)^k)\right]\nonumber\\&+\beta^2\sum_{k\in \mathcal{K}_{q,y}}\mathbb{E}\|\frac{1}{\mathcal{C}^{s}_{q,y}}\sum_{j\in N^{s}_{q,y}}\sum_{t=1}^T \nabla_y g^k_j(x_{q}^j,y_{q,t}^j,\xi_t^j)\|^2
    \nonumber
    \\
    \leq& (1-\beta\mu_g)\sum_{k\in \mathcal{K}_{q,y}}\left\|y_{q}^k-y^*(x_q)^k\right\|^2+\frac{4\beta^2NT(\sigma_{g}^2+\delta_{g}^2)}{\mathcal{C}^{*}_{y}}\nonumber\\&+\frac{2\beta NT(2\beta l_{g,1}-1)}{\mathcal{C}^*} \sum_{k\in \mathcal{K}_{q,y}}\mathbb{E}\left[g(x_q^k,y_{q,t}^k)-g(x_q^k,y^*(x_q)^k)\right]\nonumber
    \nonumber
\end{align}
Let $2\beta l_{g,1}-1\Rightarrow \beta\leq\frac{1}{2l_{g,1}}$ , we can get
\begin{align*}
     \mathbb{E}&\|y_{q+1}-y^*(x_q)\|^2\\ &\leq(1-\beta\mu_g)\mathbb{E}\|y_{q}-y^*(x_q)\|^2+\frac{4NT\beta^2}{\mathcal{C}_y^*}(\sigma_{g}^2+\delta_{g}^2)
\end{align*}
\end{proof}

\begin{lemma}
    Let all assumptions hold, we can bound,
    \begin{align}
&\mathbb{E}\left[\Phi(x_{q+1})\right]-\mathbb{E}\left[\Phi(x_q)\right]\\
&\leq-\frac{\alpha}{2}\sum_{s\in \mathcal{K}_{q,x}}\mathbb{E}\|\nabla \Phi^s(x_{q})\|^2+\frac{\alpha N}{\mathcal{C}_x^*}(2-\alpha+\alpha L_f)\sigma_f^2\nonumber\\
&\quad+\frac{\alpha N}{\mathcal{C}_x^*}(1+L_f)\delta^2_f+\frac{2\alpha N}{\mathcal{C}_x^*}(1-\alpha+\alpha L_f)\mathcal{E}_h\nonumber\\
&\quad+(-\alpha^3+\alpha^3 L_f)\frac{N}{\mathcal{C}_x^*}\sum_{s\in \mathcal{K}_{q,x}}\mathbb{E}\|\nabla f^s(x_{q},y_{q+1})\|^2\nonumber\\
&\quad+2\alpha M_f^2\mathbb{E}\|y_{q+1}-y^*(x_q)\|^2\nonumber
\end{align}
\end{lemma}
\begin{proof}
\begin{align*}
    & \mathbb{E}\left[\Phi(x_{q+1})\right]-\mathbb{E}\left[\Phi(x_q)\right] \nonumber\\
    \leq & \mathbb{E}\left[\left\langle x_{q+1} - x_q, \nabla \Phi(x_q) \right\rangle\right] +\frac{L_f}{2} \mathbb{E}\left\|x_{q+1} -x_q\right\|^2\\
    =&\sum_{s\in \mathcal{K}_{q,x}}\mathbb{E}[\langle x_{q+1}^s-x_{q}^s, \nabla \Phi^s(x_{q})\rangle]\nonumber\\
&+\sum_{s\in \mathcal{S}_x-\mathcal{K}_{q,x}}\mathbb{E}[\langle x_{q+1}^s-x_{q}^s,\nabla \Phi^s(x_{q})\rangle]+\frac{L_f}{2} \mathbb{E}\left\|x_{q+1} -x_q\right\|^2\\
    =&\sum_{s\in \mathcal{K}_{q,x}}\mathbb{E}[\langle x_{q+1}^s-x_{q}^s, \nabla \Phi^s(x_{q})\rangle]+\frac{L_f}{2} \sum_{s\in \mathcal{K}_{q,x}}\mathbb{E}\left\|x_{q+1}^s -x_q^s\right\|^2\nonumber\\
&+\frac{L_f}{2}\sum_{s\in \mathcal{S}_x-\mathcal{K}_{q,x}}\mathbb{E}\left\|x_{q+1}^s -x_q^s\right\|^2\\
    =&\sum_{s\in \mathcal{K}_{q,x}} \mathbb{E} \bigg[\Big\langle -\alpha\frac{1}{\mathcal{C}^{s}_{q,x}}\sum_{n\in N_{q,x}^{s}}\nabla f^s_n(x_{q}^n,y_{q+1}^n,\zeta_q^n),\nabla \Phi^s(x_{q})\Big\rangle\bigg]\nonumber\\
&+\frac{L_f}{2} \sum_{s\in \mathcal{K}_{q,x}}\mathbb{E}\left\|x_{q+1}^s -x_q^s\right\|^2\\
    =&(-\frac{1}{2\alpha}+\frac{L_f}{2})\underbrace{\sum_{s\in \mathcal{K}_{q,x}}\mathbb{E}\|x_{q+1}^s -x_q^s\|^2}_{U_1}-\frac{\alpha}{2}\sum_{s\in \mathcal{K}_{q,x}}\mathbb{E}\|\nabla \Phi^s(x_{q})\|^2\\
    &+\underbrace{\frac{\alpha}{2}\sum_{s\in \mathcal{K}_{q,x}}\mathbb{E}\|\frac{1}{\mathcal{C}^{s}_{q,x}}\sum_{n\in N^{s}_{q,x}}\nabla f^s_n(x_{q}^n,y_{q+1}^n,\zeta_q^n)-\nabla \Phi^s(x_{q})\|^2}_{U_2} 
\end{align*}
To bound $U_1$:
\begin{align*}
    &\sum_{s\in \mathcal{K}_{q,x}}\mathbb{E}\|x_{q+1}^s -x_q^s\|^2\\
    =&\sum_{s\in \mathcal{K}_{q,x}}\mathbb{E}\|-\alpha\frac{1}{\mathcal{C}^{s}_{q,x}}\sum_{n\in N_{q,x}^{s}}\nabla f^s_n(x_{q}^n,y_{q+1}^n,\zeta_q^n)\|^2\\
    =&\alpha^2\sum_{s\in \mathcal{K}_{q,x}}\mathbb{E}\|\frac{1}{\mathcal{C}^{s}_{q,x}}\sum_{n\in N^{s}_{q,x}}\nabla f^s_n(x_{q}^n,y_{q+1}^n,\zeta_q^n)\|^2\\
    = &\alpha^2\sum_{s\in \mathcal{K}_{q,x}}\mathbb{E}\|\frac{1}{\mathcal{C}^{s}_{q,x}}\sum_{n\in N^{s}_{q,x}}(\nabla f^s_n(x_{q}^n,y_{q+1}^n,\zeta_q^n)-\nabla f^s_n(x_{q}^n,y_{q+1}^n)\\
    &+\nabla f^s_n(x_{q}^n,y_{q+1}^n)-\nabla f^s_n(x_{q},y_{q+1})+\nabla f^s_n(x_{q},y_{q+1})\\
    &-\nabla f^s(x_{q},y_{q+1})+\nabla f^s(x_{q},y_{q+1}))\|^2 \\
    \leq&4\alpha^2\frac{N}{\mathcal{C}_x^*}\sigma_f^2+4\alpha^2\frac{N}{\mathcal{C}_x^*}\delta^2_f+4\alpha^2\frac{N}{\mathcal{C}_x^*}\mathcal{E}_h\\&+4\alpha^2\frac{N}{\mathcal{C}_x^*}\sum_{s\in \mathcal{K}_{q,x}}\mathbb{E}\|\nabla f^s(x_{q},y_{q+1})\|^2
\end{align*}
To bound $U_2$:
\begin{align*}
    &\frac{\alpha}{2}\sum_{s\in \mathcal{K}_{q,x}}\mathbb{E}\|\frac{1}{\mathcal{C}^{s}_{q,x}}\sum_{n\in N^{s}_{q,x}}\nabla f^s_n(x_{q}^n,y_{q+1}^n,\zeta_q^n)-\nabla \Phi^s(x_{q})\|^2\\
    \leq &2\alpha \frac{N}{\mathcal{C}_x^*}\sigma_f^2+2\alpha\frac{N}{\mathcal{C}_x^*}\mathcal{E}_h+2\alpha\frac{N}{\mathcal{C}_x^*}\delta^2_f\\&+ 2\alpha \sum_{s\in \mathcal{K}_{q,x}}\mathbb{E}\|\nabla f^s(x_{q},y_{q+1})-\nabla \Phi^s(x_{q})\|^2\\
    \leq &2\alpha \frac{N}{\mathcal{C}_x^*}\sigma_f^2+2\alpha\frac{N}{\mathcal{C}_x^*}\mathcal{E}_h+2\alpha\frac{N}{\mathcal{C}_x^*}\delta^2_f+2\alpha M_f^2 \mathbb{E}\|y_{q+1}-y^*(x_q)\|^2
\end{align*}
Thus, we can get
\begin{align*}
&\mathbb{E}\left[\Phi(x_{q+1})\right]-\mathbb{E}\left[\Phi(x_q)\right]
\\\leq&(-\frac{1}{2\alpha}+\frac{L_f}{2})(4\alpha^2\frac{N}{\mathcal{C}_x^*}\sigma_f^2+4\alpha^2\frac{N}{\mathcal{C}_x^*}\delta^2_f+4\alpha^2\frac{N}{\mathcal{C}_x^*}\mathcal{E}_h\\&+4\alpha^2\frac{N}{\mathcal{C}_x^*}\sum_{s\in \mathcal{K}_{q,x}}\mathbb{E}\|\nabla f^s(x_{q},y_{q+1})\|^2)\\
& -\frac{\alpha}{2}\sum_{s\in \mathcal{K}_{q,x}}\mathbb{E}\|\nabla \Phi^s(x_{q})\|^2+2\alpha \frac{N}{\mathcal{C}_x^*}\sigma_f^2+2\alpha\frac{N}{\mathcal{C}_x^*}\mathcal{E}_h\\&+2\alpha\frac{N}{\mathcal{C}_x^*}\delta^2_f+2\alpha M_f^2\mathbb{E}\|y_{q+1}-y^*(x_q)\|^2\\ 
=&-\frac{\alpha}{2}\sum_{s\in \mathcal{K}_{q,x}}\mathbb{E}\|\nabla \Phi^s(x_{q})\|^2+\frac{2L_f\alpha^2 N}{\mathcal{C}_x^*}(\sigma_f^2+\delta^2_f+\mathcal{E}_h)\\
&+(-2\alpha+2 L_f\alpha^2)\frac{N}{\mathcal{C}_x^*}\sum_{s\in \mathcal{K}_{q,x}}\mathbb{E}\|\nabla f^s(x_{q},y_{q+1})\|^2\\&+2\alpha M_f^2\mathbb{E}\|y_{q+1}-y^*(x_q)\|^2
\end{align*}
\end{proof}

\begin{lemma}
    Let all assumptions hold, we can gain,
    \begin{align}
    &\mathbb{E}\|y_{q+1}-y^*(x_{q+1})\|^2\\
    &\leq(2-2\beta\mu_g)\mathbb{E}\|y_{q}-y^*(x_q)\|^2+\frac{8NT\beta^2}{\mathcal{C}_y^*}(\sigma_{g}^2+\delta_{g}^2)\nonumber\\&+ 8L_y\alpha^2\frac{N}{\mathcal{C}_x^*}\sigma_f^2+8L_y\alpha\frac{N}{\mathcal{C}_x^*}\delta^2_f\nonumber\\
    &+8L_y\alpha^2\frac{N}{\mathcal{C}_x^*}\mathcal{E}_h+8L_y\alpha^3\frac{N}{\mathcal{C}_x^*}\sum_{s\in \mathcal{K}_{q,x}}\mathbb{E}\|\nabla f^s(x_{q},y_{q+1})\|^2\nonumber
    \end{align}
\end{lemma}
\begin{proof}
\begin{align*}
    &\mathbb{E}\|y_{q+1}-y^*(x_{q+1})\|^2\\
    \leq &2\mathbb{E}\|y_{q+1}-y^*(x_{q})\|^2+2\mathbb{E}\|y^*(x_{q})-y^*(x_{q+1})\|^2\\
    \leq&2\mathbb{E}\|y_{q+1}-y^*(x_{q})\|^2+2L_y\mathbb{E}\|x_{q+1}-x_{q}\|^2\\
    =&2\mathbb{E}\|y_{q+1}-y^*(x_{q})\|^2+2L_y\sum_{s\in \mathcal{K}_{q,x}}\mathbb{E}\|x_{q+1}-x_{q}\|^2\\&+2L_y\sum_{s\in \mathcal{S}_x-\mathcal{K}_{q,x}}\mathbb{E}\|x_{q+1}-x_{q}\|^2 \\
    =&2\mathbb{E}\|y_{q+1}-y^*(x_{q})\|^2+2L_y\sum_{s\in \mathcal{K}_{q,x}}\mathbb{E}\|x_{q+1}-x_{q}\|^2\\
    \leq&(2-2\beta\mu_g)\mathbb{E}\|y_{q}-y^*(x_q)\|^2+\frac{8NT\beta^2}{\mathcal{C}_y^*}(\sigma_{g}^2+\delta_{g}^2)\\&+ 8L_y\alpha^2\frac{N}{\mathcal{C}_x^*}\sigma_f^2+8L_y\alpha\frac{N}{\mathcal{C}_x^*}\delta^2_f+8L_y\alpha^2\frac{N}{\mathcal{C}_x^*}\mathcal{E}_h\\
    &+8L_y\alpha^3\frac{N}{\mathcal{C}_x^*}\sum_{s\in \mathcal{K}_{q,x}}\mathbb{E}\|\nabla f^s(x_{q},y_{q+1})\|^2
\end{align*}
\end{proof}

\textbf{Proof of Theorem~\ref{convergence1}}
\begin{proof}
We define the following Lyapunov function
\begin{align*}
    V_q:=\Phi(x_q)+\frac{M_f}{L_y}\|y_q-y^*(x_q)\|^2
\end{align*}
The bound of two Lyapunov functions is  
\begin{align*}
&V_{q+1}-V_{q}\\&=\Phi(x_{q+1})-\Phi(x_q)\\&\quad+\frac{M_f}{L_y}\left( \|y_{q+1}-y^*(x_{q+1})\|^2-\|y_q-y^*(x_q)\|^2\right)
\end{align*}  
From those lemmas, we can get  
\begin{subequations}
\begin{align}
    &\mathbb{E}\| V_{q+1}\|-\mathbb{E}\|V_{q} \|\nonumber\\&\leq-\frac{\alpha}{2}\sum_{s\in \mathcal{K}_{q,x}}\mathbb{E}\|\nabla \Phi^s(x_{q})\|^2+\frac{2L_f\alpha^2 N}{\mathcal{C}_x^*}(\sigma_f^2+\delta^2_f+\mathcal{E}_h)\nonumber\\
    &+(-2\alpha+2 L_f\alpha^2)\frac{N}{\mathcal{C}_x^*}\sum_{s\in \mathcal{K}_{q,x}}\mathbb{E}\|\nabla f^s(x_{q},y_{q+1})\|^2\nonumber\\
    &+(2\alpha M_f^2+2\frac{M_f}{L_y})\mathbb{E}\|y_{q+1}-y^*(x_q)\|^2-\frac{M_f}{L_y}\|y_q-y^*(x_q)\|^2\nonumber\\
    &+8\alpha^2M_f\frac{N}{\mathcal{C}_x^*}(\sigma_f^2+\delta^2_f+\mathcal{E}_h+\sum_{s\in \mathcal{K}_{q,x}}\mathbb{E}\|\nabla f^s(x_{q},y_{q+1})\|^2)\nonumber\\
    &\leq -\frac{\alpha}{2}\sum_{s\in \mathcal{K}_{q,x}}\mathbb{E}\|\nabla \Phi^s(x_{q})\|^2+\frac{\hat{\alpha}_1 N}{\mathcal{C}_x^*}(\sigma_f^2+\delta^2_f+\mathcal{E}_h)\nonumber\\
    &+(-2\alpha+2 L_f\alpha^2+8\alpha^2M_f)\frac{N}{\mathcal{C}_x^*}\sum_{s\in \mathcal{K}_{q,x}}\mathbb{E}\|\nabla f^s(x_{q},y_{q+1})\|^2\label{eq:lya_a_rabo}\\
    &+(2\alpha M_f^2+2\frac{M_f}{L_y})\mathbb{E}\|y_{q+1}-y^*(x_q)\|^2-\frac{M_f}{L_y}\|y_q-y^*(x_q)\|^2\label{eq:lya_b_rabo}
\end{align}
\end{subequations}
where $\hat{\alpha}_1=2L_f\alpha^2+8M_f\alpha^2$.

It's noticed that Equation~$(\ref{eq:lya_a_rabo})\leq 0$ if $\alpha\leq\frac{1}{L_f+4M_f}$ and we enforce $\alpha$ satisfies it in the following context.

Based on Lemma~\ref{lemma:error_inner}, Equation~(\ref{eq:lya_b_rabo}) can be bounded as
\begin{align}
    &Equation~(\ref{eq:lya_b_rabo})\nonumber\\&\leq(2\alpha M_f^2+2\frac{M_f}{L_y})\nonumber
    \\
    &\quad\cdot((1-\beta\mu_g)\sum_{k\in \mathcal{K}_{q,y}}\left\|y_{q}^k-y^*(x_q)^k\right\|^2 +\frac{4\beta^2NT(\sigma_{g}^2+\delta_{g}^2)}{\mathcal{C}^{*}_{y}})\nonumber
    \\
    &\quad-\frac{M_f}{L_y}\|y_q-y^*(x_q)\|^2
    \nonumber
    \\
    &= \frac{M_f}{L_y}\left((2\alpha L_y M_f+2)(1-\beta \mu_g)-1\right)\sum_{k\in \mathcal{K}_{q,y}}\left\|y_{q}^k-y^*(x_q)^k\right\|^2\nonumber
    \\
    &\quad+(2\alpha M_f^2+2\frac{M_f}{L_y})\frac{4\beta^2NT(\sigma_{g}^2+\delta_{g}^2)}{\mathcal{C}^{*}_{y}}
    \nonumber
\end{align}
If $\beta\leq\frac{1}{\mu_g}$ and $\beta \geq \frac{1}{\mu_g}-\frac{1}{2\alpha L_yM_f\mu_g}$, we can get  
\begin{align*}
\frac{M_f}{L_y}\left((2\alpha L_y M_f+2)(1-\beta \mu_g)-1\right)\leq0
\end{align*}
After rearranging, we can gain that
\begin{align*}
    &\mathbb{E}\| V_{q+1}\|-\mathbb{E}\|V_{q} \|\nonumber\\&\leq-\frac{\alpha}{2}\sum_{s\in \mathcal{K}_{q,x}}\mathbb{E}\|\nabla \Phi^s(x_{q})\|^2+\frac{\hat{\alpha}_1N}{\mathcal{C}_x^*}\sigma_f^2+\frac{\hat{\alpha}_1N}{\mathcal{C}_x^*}\delta^2_f\\
    &\quad+\frac{\hat{\alpha}_1N}{\mathcal{C}_x^*}(2M_f^2 w_1^2 \|x_q\|^2+2M_f^2 w_2^2 \|y_{q+1}\|^2)\\
    &\quad+(2\alpha M_f^2+2\frac{M_f}{L_y})\frac{4\beta^2NT(\sigma_{g}^2+\delta_{g}^2)}{\mathcal{C}^{*}_{y}}
\end{align*}
Summing the above result on $Q$ communication rounds, we can gain that 
\begin{align*}
     &\frac{1}{Q}\sum_{q=1}^Q\sum_{s\in \mathcal{K}_{q,x}}\mathbb{E}\|\nabla \Phi^s(x_{q})\|^2\\&\leq \frac{2\Delta}{\alpha Q}+\frac{2N\hat{\alpha}_1}{\alpha\mathcal{C}_x^*}\sigma_f^2+\frac{2N\hat{\alpha}_1}{\alpha\mathcal{C}_x^*}\delta^2_f\\&\quad+\frac{4N\hat{\alpha}_1}{\alpha\mathcal{C}_x^*}\frac{1}{Q}\sum_{q=1}^Q(M_f^2 w_1^2 \|x_q\|^2+M_f^2 w_2^2 \|y_{q+1}\|^2)\\
     &\quad+(4M_f^2+\frac{4M_f}{\alpha L_y})\frac{4\beta^2NT(\sigma_{g}^2+\delta_{g}^2)}{\mathcal{C}^{*}_{y}}
\end{align*}
where $\Delta=\mathbb{E}\|V_{Q+1}\|-\mathbb{E}\| V_{1}\|$.

Here, we enforce $\alpha\leq \sqrt{\frac{C_x^*}{Q}}$ for some positive $\hat{\alpha}$, and $\beta \leq \frac{1}{NTQ}$, we can get
\begin{align*}
     &\frac{1}{Q}\sum_{q=1}^Q\sum_{s\in \mathcal{K}_{q,x}}\mathbb{E}\|\nabla \Phi^s(x_{q})\|^2\\\leq& \frac{2\Delta}{\sqrt{\mathcal{C}_x^*Q}}+\frac{4N(L_f+4M_f)}{\sqrt{\mathcal{C}_x^*Q}}\sigma_f^2+\frac{4N(L_f+4M_f)}{\sqrt{\mathcal{C}_x^*Q}}\delta^2_f\\&+\frac{8N(L_f+4M_f)}{\sqrt{\mathcal{C}_x^*Q}}\frac{1}{Q}\sum_{q=1}^Q(M_f^2 w_2^2\|x_q\|^2+M_f^2 w_2^2 \|y_{q+1}\|^2)\\
     &+\frac{16M_f^2(\sigma_{g}^2+\delta_{g}^2)}{NTQ^2\mathcal{C}^{*}_{y}}+\frac{16M_f(\sigma_{g}^2+\delta_{g}^2)}{NTQ^{1.5}L_y\sqrt{\mathcal{C}_x^*}\mathcal{C}^{*}_{y}}
\end{align*}
Thus, we can gain 
\begin{align*}
    &\frac{1}{Q}\sum_{q=1}^Q\sum_{s\in \mathcal{K}_{q,x}}\mathbb{E}\|\nabla \Phi^s(x_{q})\|^2\\&\leq \mathcal{O}(\frac{1}{\sqrt{\mathcal{C}_x^*Q}}+\frac{1}{NTQ^2\mathcal{C}^{*}_{y}}+\frac{1}{NTQ^{1.5}\sqrt{\mathcal{C}_x^*}\mathcal{C}^{*}_{y}})
\end{align*}
\end{proof}

\subsection{Convergence Analysis of RAFBO}
\begin{lemma}
    Let all assumptions hold, we can bound,
    \begin{align}
&\mathbb{E}\left[\Phi(x_{q+1})\right]-\mathbb{E}\left[\Phi(x_q)\right]\\
&\leq-\frac{\alpha}{2}\sum_{s\in \mathcal{K}_{q,x}}\mathbb{E}\|\nabla \Phi^s(x_{q})\|^2+\frac{\alpha N}{\mathcal{C}_x^*}(2-\alpha+\alpha L_f)\sigma_f^2\nonumber\\
&\quad+\frac{\alpha N}{\mathcal{C}_x^*}(1+L_f)\delta^2_f+\frac{2\alpha N}{\mathcal{C}_x^*}(1-\alpha+\alpha L_f)\mathcal{E}_h\nonumber\\
&\quad+(-\alpha^3+\alpha^3 L_f)\frac{N}{\mathcal{C}_x^*}\sum_{s\in \mathcal{K}_{q,x}}\mathbb{E}\|\nabla f^s(x_{q},y_{q+1})\|^2\nonumber\\
&\quad+2\alpha M_f^2\mathbb{E}\|y_{q+1}-y^*(x_q)\|^2\nonumber
\end{align}
\end{lemma}
\begin{proof}
\begin{align*}
    & \mathbb{E}\left[\Phi(x_{q+1})\right]-\mathbb{E}\left[\Phi(x_q)\right] \nonumber\\
    \leq & \mathbb{E}\left[\left\langle x_{q+1} - x_q, \nabla \Phi(x_q) \right\rangle\right] +\frac{L_f}{2} \mathbb{E}\left\|x_{q+1} -x_q\right\|^2\\
    =&(-\frac{1}{2\alpha}+\frac{L_f}{2})\underbrace{\sum_{s\in \mathcal{K}_{q,x}}\mathbb{E}\|x_{q+1}^s -x_q^s\|^2}_{U_1'}-\frac{\alpha}{2}\sum_{s\in \mathcal{K}_{q,x}}\mathbb{E}\|\nabla \Phi^s(x_{q})\|^2\\
    &+\underbrace{\frac{\alpha}{2}\sum_{s\in \mathcal{K}_{q,x}}\mathbb{E}\|\frac{1}{\mathcal{C}^{s}_{q,x}}\sum_{n\in N^{s}_{q,x}}\nabla f^s_n(x_{q}^n,y_{q+1}^n,\zeta_q^n)-\nabla \Phi^s(x_{q})\|^2}_{U_2'} 
\end{align*}
To bound $U_1'$:
\begin{align*}
    &\sum_{s\in \mathcal{K}_{q,x}}\mathbb{E}\|x_{q+1}^s -x_q^s\|^2\\
    =&\sum_{s\in \mathcal{K}_{q,x}}\mathbb{E}\|-\alpha\frac{1}{\mathcal{C}^{s}_{q,x}}\sum_{n\in N_{q,x}^{s}}\hat{\nabla} f^s_n(x_{q}^n,y_{q+1}^n,\zeta_q^n)\|^2\\
    \leq&5\alpha^2\frac{N}{\mathcal{C}_x^*}\sigma_{\mathcal{J}}^2+5\alpha^2\frac{N}{\mathcal{C}_x^*}\sigma_f^2+5\alpha^2\frac{N}{\mathcal{C}_x^*}\delta^2_f+5\alpha^2\frac{N}{\mathcal{C}_x^*}\mathcal{E}_h\\&+5\alpha^2\frac{N}{\mathcal{C}_x^*}\sum_{s\in \mathcal{K}_{q,x}}\mathbb{E}\|\nabla f^s(x_{q},y_{q+1})\|^2
\end{align*}
To bound $U_2'$:

\begin{align*}
    &\frac{\alpha}{2}\sum_{s\in \mathcal{K}_{q,x}}\mathbb{E}\|\frac{1}{\mathcal{C}^{s}_{q,x}}\sum_{n\in N^{s}_{q,x}}\nabla f^s_n(x_{q}^n,y_{q+1}^n,\zeta_q^n)-\nabla \Phi^s(x_{q})\|^2\\
    \leq &\frac{5}{2}\alpha \frac{N}{\mathcal{C}_x^*}\sigma_{\mathcal{J}}^2+\frac{5}{2}\alpha \frac{N}{\mathcal{C}_x^*}\sigma_f^2+\frac{5}{2}\alpha\frac{N}{\mathcal{C}_x^*}\mathcal{E}_h+\frac{5}{2}\alpha\frac{N}{\mathcal{C}_x^*}\delta^2_f\\&+\frac{5}{2}\alpha M_f^2 \mathbb{E}\|y_{q+1}-y^*(x_q)\|^2
\end{align*}
Thus, we can get
\begin{align*}
&\mathbb{E}\left[\Phi(x_{q+1})\right]-\mathbb{E}\left[\Phi(x_q)\right]\\
&\leq(-\frac{1}{2\alpha}+\frac{ L_f}{2})(5\alpha^2\frac{N}{\mathcal{C}_x^*}\sigma_{\mathcal{J}}^2+5\alpha^2\frac{N}{\mathcal{C}_x^*}\sigma_f^2+5\alpha^2\frac{N}{\mathcal{C}_x^*}\delta^2_f\\&\quad+5\alpha^2\frac{N}{\mathcal{C}_x^*}\mathcal{E}_h+5\alpha^2\frac{N}{\mathcal{C}_x^*}\sum_{s\in \mathcal{K}_{q,x}}\mathbb{E}\|\nabla f^s(x_{q},y_{q+1})\|^2)\\
&\quad -\frac{\alpha}{2}\sum_{s\in \mathcal{K}_{q,x}}\mathbb{E}\|\nabla \Phi^s(x_{q})\|^2+\frac{5}{2}\alpha \frac{N}{\mathcal{C}_x^*}\sigma_{\mathcal{J}}^2+\frac{5}{2}\alpha \frac{N}{\mathcal{C}_x^*}\sigma_f^2\\&\quad+\frac{5}{2}\alpha\frac{N}{\mathcal{C}_x^*}\mathcal{E}_h+\frac{5}{2}\alpha\frac{N}{\mathcal{C}_x^*}\delta^2_f+\frac{5}{2}\alpha M_f^2 \mathbb{E}\|y_{q+1}-y^*(x_q)\|^2\\ 
&=-\frac{\alpha}{2}\sum_{s\in \mathcal{K}_{q,x}}\mathbb{E}\|\nabla \Phi^s(x_{q})\|^2+\frac{5L_f\alpha^2 N}{2\mathcal{C}_x^*}(\sigma_{\mathcal{J}}^2+\sigma_f^2+\delta^2_f+\mathcal{E}_h)\\
&\quad+(-\frac{5}{2}\alpha+\frac{5}{2} L_f\alpha^2)\frac{N}{\mathcal{C}_x^*}\sum_{s\in \mathcal{K}_{q,x}}\mathbb{E}\|\nabla f^s(x_{q},y_{q+1})\|^2\\&\quad+\frac{5}{2}\alpha M_f^2\mathbb{E}\|y_{q+1}-y^*(x_q)\|^2
\end{align*}
\end{proof}

\begin{lemma}
    Let all assumptions hold, we can gain,
    \begin{align}
    &\mathbb{E}\|y_{q+1}-y^*(x_{q+1})\|^2\\
    \leq&(2-2\beta\mu_g)\mathbb{E}\|y_{q}-y^*(x_q)\|^2+\frac{8NT\beta^2}{\mathcal{C}_y^*}(\sigma_{g}^2+\delta_{g}^2)\nonumber\\
    &+ 10L_y\alpha^2\frac{N}{\mathcal{C}_x^*}\sigma_{\mathcal{J}}^2+10L_y\alpha^2\frac{N}{\mathcal{C}_x^*}\sigma_f^2+10L_y\alpha\frac{N}{\mathcal{C}_x^*}\delta^2_f\nonumber\\
    &+10L_y\alpha^2\frac{N}{\mathcal{C}_x^*}\mathcal{E}_h+10L_y\alpha^3\frac{N}{\mathcal{C}_x^*}\sum_{s\in \mathcal{K}_{q,x}}\mathbb{E}\|\nabla f^s(x_{q},y_{q+1})\|^2\nonumber
    \end{align}
\end{lemma}

\textbf{Proof of Theorem~\ref{convergence2}}
\begin{proof}
We define the following Lyapunov function
\begin{align*}
    V_q:=\Phi(x_q)+\frac{M_f}{L_y}\|y_q-y^*(x_q)\|^2
\end{align*}
The bound of two Lyapunov functions is  
\begin{align*}
    &V_{q+1}-V_{q}\\&=\Phi(x_{q+1})-\Phi(x_q)\\&\quad+\frac{M_f}{L_y}\left( \|y_{q+1}-y^*(x_{q+1})\|^2-\|y_q-y^*(x_q)\|^2\right)
\end{align*}  
From those lemmas, we can get  
\begin{subequations}
\begin{align}
    &\mathbb{E}\| V_{q+1}\|-\mathbb{E}\|V_{q} \|\nonumber\\&\leq-\frac{\alpha}{2}\sum_{s\in \mathcal{K}_{q,x}}\mathbb{E}\|\nabla \Phi^s(x_{q})\|^2+\frac{5L_f\alpha^2 N}{2\mathcal{C}_x^*}(\sigma_{\mathcal{J}}^2+\sigma_f^2+\delta^2_f+\mathcal{E}_h)\nonumber\\
    &+(-\frac{5}{2}\alpha+\frac{5}{2} L_f\alpha^2)\frac{N}{\mathcal{C}_x^*}\sum_{s\in \mathcal{K}_{q,x}}\mathbb{E}\|\nabla f^s(x_{q},y_{q+1})\|^2\nonumber\\
    &+(\frac{5}{2}\alpha M_f^2+2\frac{M_f}{L_y})\mathbb{E}\|y_{q+1}-y^*(x_q)\|^2-\frac{M_f}{L_y}\|y_q-y^*(x_q)\|^2\nonumber\\
    &+10\alpha^2M_f\frac{N}{\mathcal{C}_x^*}\nonumber\\&\quad\cdot(\sigma_{\mathcal{J}}^2+\sigma_f^2+\delta^2_f+\mathcal{E}_h+\sum_{s\in \mathcal{K}_{q,x}}\mathbb{E}\|\nabla f^s(x_{q},y_{q+1})\|^2)\nonumber\\
    &\leq -\frac{\alpha}{2}\sum_{s\in \mathcal{K}_{q,x}}\mathbb{E}\|\nabla \Phi^s(x_{q})\|^2+\frac{\hat{\alpha}_2 N}{\mathcal{C}_x^*}(\sigma_{\mathcal{J}}^2+\sigma_f^2+\delta^2_f+\mathcal{E}_h)\nonumber\\
    &+(-\frac{5}{2}\alpha+\frac{5}{2} L_f\alpha^2+10\alpha^2M_f)\frac{N}{\mathcal{C}_x^*}\sum_{s\in \mathcal{K}_{q,x}}\mathbb{E}\|\nabla f^s(x_{q},y_{q+1})\|^2\label{eq:lya_a}\\
    &+(\frac{5}{2}\alpha M_f^2+2\frac{M_f}{L_y})\mathbb{E}\|y_{q+1}-y^*(x_q)\|^2-\frac{M_f}{L_y}\|y_q-y^*(x_q)\|^2\label{eq:lya_b}
\end{align}
\end{subequations}
where $\hat{\alpha}_2=\frac{5}{2}L_f\alpha^2+10M_f\alpha^2$.

It's noticed that Equation~$(\ref{eq:lya_a})\leq 0$ if $\alpha\leq\frac{1}{L_f+4M_f}$ and we enforce $\alpha$ satisfies it in the following context.

Based on Lemma~\ref{lemma:error_inner}, Equation~(\ref{eq:lya_b}) can be bounded as
\begin{align*}
    &Equation~(\ref{eq:lya_b})\nonumber\\&\leq(2\alpha M_f^2+2\frac{M_f}{L_y})\\&\quad\cdot((1-\beta\mu_g)\sum_{k\in \mathcal{K}_{q,y}}\left\|y_{q}^k-y^*(x_q)^k\right\|^2 +\frac{4\beta^2NT(\sigma_{g}^2+\delta_{g}^2)}{\mathcal{C}^{*}_{y}})\\&\quad-\frac{M_f}{L_y}\|y_q-y^*(x_q)\|^2
    \nonumber
    \\
    &= \frac{M_f}{L_y}\left((2\alpha L_y M_f+2)(1-\beta \mu_g)-1\right)\sum_{k\in \mathcal{K}_{q,y}}\left\|y_{q}^k-y^*(x_q)^k\right\|^2\\&\quad+(2\alpha M_f^2+2\frac{M_f}{L_y})\frac{4\beta^2NT(\sigma_{g}^2+\delta_{g}^2)}{\mathcal{C}^{*}_{y}}
    \nonumber
\end{align*}
If $\beta\leq\frac{1}{\mu_g}$ and $\beta \geq \frac{1}{\mu_g}-\frac{1}{2\alpha L_yM_f\mu_g}$, we can get  
\begin{align*}
\frac{M_f}{L_y}\left((2\alpha L_y M_f+2)(1-\beta \mu_g)-1\right)\leq0
\end{align*}
After rearranging, we can gain that
\begin{align*}
    &\mathbb{E}\| V_{q+1}\|-\mathbb{E}\|V_{q} \|
    \\& \leq-\frac{\alpha}{2}\sum_{s\in \mathcal{K}_{q,x}}\mathbb{E}\|\nabla \Phi^s(x_{q})\|^2+\frac{\hat{\alpha}_2N}{\mathcal{C}_x^*}\sigma_{\mathcal{J}}^2+\frac{\hat{\alpha}_2N}{\mathcal{C}_x^*}\sigma_f^2+\frac{\hat{\alpha}_2N}{\mathcal{C}_x^*}\delta^2_f\\
    &\quad+\frac{\hat{\alpha}_2N}{\mathcal{C}_x^*}(2M_f^2 w_1^2 \|x_q\|^2+2M_f^2 w_2^2 \|y_{q+1}\|^2)\\
    &\quad+(\frac{5}{2}\alpha M_f^2+2\frac{M_f}{L_y})\frac{4\beta^2NT(\sigma_{g}^2+\delta_{g}^2)}{\mathcal{C}^{*}_{y}}
\end{align*}

Summing the result of on $Q$ communication rounds, we can gain that
\begin{align*}
     \frac{1}{Q}&\sum_{q=1}^Q\sum_{s\in \mathcal{K}_{q,x}}\mathbb{E}\|\nabla \Phi^s(x_{q})\|^2\\\leq& \frac{2\Delta}{\alpha Q}+\frac{4N\hat{\alpha}_1}{\alpha\mathcal{C}_x^*}\frac{1}{Q}\sum_{q=1}^Q(M_f^2 w_1^2 \|x_q\|^2+M_f^2 w_2^2 \|y_{q+1}\|^2)\\&+(4M_f^2+\frac{4M_f}{\alpha L_y})\frac{4\beta^2NT(\sigma_{g}^2+\delta_{g}^2)}{\mathcal{C}^{*}_{y}}\\
     &+\frac{2N\hat{\alpha}_1}{\alpha\mathcal{C}_x^*}\sigma_{\mathcal{J}}^2+\frac{2N\hat{\alpha}_1}{\alpha\mathcal{C}_x^*}\sigma_f^2+\frac{2N\hat{\alpha}_1}{\alpha\mathcal{C}_x^*}\delta^2_f
\end{align*}
where $\Delta=\mathbb{E}\|V_{Q+1}\|-\mathbb{E}\| V_{1}\|$.

Here, we enforce $\alpha\leq \sqrt{\frac{C_x^*}{Q}}$ for some positive $\hat{\alpha}$, and $\beta \leq \frac{1}{NTQ}$, we can get
\begin{align*}
     \frac{1}{Q}&\sum_{q=1}^Q\sum_{s\in \mathcal{K}_{q,x}}\mathbb{E}\|\nabla \Phi^s(x_{q})\|^2\\\leq& \frac{2\Delta}{\sqrt{\mathcal{C}_x^*Q}}+\frac{4N(L_f+4M_f)}{\sqrt{\mathcal{C}_x^*Q}}\sigma_{\mathcal{J}}^2+\frac{4N(L_f+4M_f)}{\sqrt{\mathcal{C}_x^*Q}}\sigma_f^2\\&+\frac{4N(L_f+4M_f)}{\sqrt{\mathcal{C}_x^*Q}}\delta^2_f\\&+\frac{8N(L_f+4M_f)}{\sqrt{\mathcal{C}_x^*Q}}\frac{1}{Q}\sum_{q=1}^Q(M_f^2 w_2^2\|x_q\|^2+M_f^2 w_2^2 \|y_{q+1}\|^2)\\
     &+\frac{16M_f^2(\sigma_{g}^2+\delta_{g}^2)}{NTQ^2\mathcal{C}^{*}_{y}}+\frac{16M_f(\sigma_{g}^2+\delta_{g}^2)}{NTQ^{1.5}L_y\sqrt{\mathcal{C}_x^*}\mathcal{C}^{*}_{y}}
\end{align*}
Thus, we can gain 
\begin{align*}
    &\frac{1}{Q}\sum_{q=1}^Q\sum_{s\in \mathcal{K}_{q,x}}\mathbb{E}\|\nabla \Phi^s(x_{q})\|^2\\&\leq \mathcal{O}(\frac{1}{\sqrt{\mathcal{C}_x^*Q}}+\frac{1}{NTQ^2\mathcal{C}^{*}_{y}}+\frac{1}{NTQ^{1.5}\sqrt{\mathcal{C}_x^*}\mathcal{C}^{*}_{y}})
\end{align*}
\end{proof}

\section{Experiments}\label{sec:exp}
In this section, we assess the effectiveness of the proposed method on two distributed bilevel tasks: hyper-representation problem for few-shot learning and hyper-parameter optimization problem for loss function tuning.

\subsection{Hyper-representation for Few-shot Learning}\label{sec:few-shot}
In the hyper-representation learning task, we learn a hyper-representation of the data such that a linear classifier can be learned quickly with a small number of data samples.
We consider the following practical optimization problem,

\begin{center}
\scalebox{0.94}{$
\begin{aligned}
&\min_{x} \Phi(x)= 
\frac{1}{n} \sum_{i=1}^n\left(\frac{1}{N_m} \sum_{j=1}^{N_m} \mathcal{L}_{\mathcal{Q}_j}(x, y^*(x))\right) \nonumber
\\& \;\;\mbox{s.t.} \quad y^*(x)\in \arg\min_{y} 
\frac{1}{n}\sum_{i=1}^n\left(\frac{1}{N_m} \sum_{j=1}^{N_m} \mathcal{L}_{\mathcal{S}_j}(x,y))\right).
\end{aligned}$}  
\end{center}

In above formulation, we set $n$ clients and each client owns $N_m$ tasks, and each task concludes a support set and a query set. The function
$\mathcal{L}_{\mathcal{S}_j}(x, y)=\frac{1}{|\mathcal{S}_j|}\sum_{\xi\in\mathcal{S}_j}\mathcal{L}(x,y,\xi)$ and 
$\mathcal{L}_{\mathcal{Q}_j}(x, y^*(x))=\frac{1}{|\mathcal{Q}_j|}\sum_{\zeta\in\mathcal{Q}_j}\mathcal{L}(x,y^*(x),\zeta)$ correspond respectively to the support and query loss functions for local task $j$.
We use the Resnet12 backbone as $x$ for feature extraction and a MLP as $y$ for classify. Thus, the lower level problem is to learn the optimal linear classifier $y$ given the backbone $x$, and the upper level problem is to learn the optimal backbone parameter $x$.

\begin{table}[t]
    \centering
    \caption{Experimental results on N-way-K-shot classification on Omniglot dataset.}
    \label{tab:baselines}
    \begin{tabular}{lcccc} 
    \toprule
    \multirow{2}{*}{Algorithm} & \multicolumn{2}{c}{5-way Accuracy}& \multicolumn{2}{c}{20-way Accuracy}\\ 
    \cline{2-5}
    & 1-shot & 5-shot & 1-shot & 5-shot \\ 
    \midrule
    \rowcolor{mypurple} \multicolumn{5}{c}{\textit{\textcolor{gray}{Full model}}} \\
    Dist-ProtoNet & 80.76\% & 89.33\% & 60.05\% & 78.01\%\\
    Dist-MAML & 81.46\% & 92.01\% & 61.36\% & 78.35\%\\
    FedMBO\cite{conf/icml/fedmbo} & 82.66\% & 93.69\% & 63.18\% & 81.64\%\\
    \midrule
    \rowcolor{mypurple} \multicolumn{5}{c}{\textit{\textcolor{gray}{Sub model}}} \\
    Static & 72.34\% & 85.09\% & 42.99\% & 51.01\%\\
    OAP\cite{zhou2024every} & 73.58\% & 86.16\% & 45.45\% &51.92\%\\
    FedRolex\cite{alam2022fedrolex} & 73.74\% & 86.95\% & 47.25\% & 55.29\%\\
    FIARSE\cite{wu2024fiarse} & 74.33\% & 87.78\% & 49.52\% & 55.87\%\\
    \midrule
    \rowcolor{mypurple}RABO(Ours)& \textbf{76.05\%} & \textbf{89.12\%} & \textbf{54.97\%} & \textbf{77.41\%}\\ 
    \rowcolor{mypurple}RAFBO(Ours)& \textbf{75.09\% }& \textbf{88.57\%} & \textbf{51.37\%} & \textbf{71.16\%}\\
    \bottomrule
    \end{tabular}
\end{table}

We conduct few-shot meta-learning experiments on the Omniglot dataset~\cite{lake2015human}. The dataset comprises 1623 characters from 50 different alphabets, with each character containing 20 samples. We first partition the alphabets into training and testing sets with a 30/20 split. The characters in the training dataset are further divided into 10 shards, with each shard assigned to a client as its local dataset. We perform N-way-K-shot classification tasks. For each task, we randomly sample N characters, and for each character, we select K samples as the support set and 20-K samples as the query set.

We compare our proposed methods, \textit{RABO} and \textit{RAFBO}, with traditional distributed few-shot learning approaches. For full-model training comparisons, we select the following baselines: ProtoNet~\cite{snell2017prototypical}, a traditional few-shot learning algorithm; MAML~\cite{finn2017model}, a meta-training algorithm; and FedMBO~\cite{conf/icml/fedmbo}, a distributed bilevel optimization algorithm. For single-level submodel training, we include Static, OAP~\cite{zhou2024every}, FedRolex~\cite{alam2022fedrolex}, and FIARSE~\cite{wu2024fiarse} as baselines. We consider five different capacity settings: $\beta = \{1,1/2,1/4,1/8,1/16\}$. For instance, 1/2 means the client model capacity is half of the full model. For \textit{RABO} and \textit{RAFBO}, we first split the model according to the weight importance and each client selects a submodel that matches its capacity. And both feature extractor $x$ and classifier $y$ are followed the same model capacities setting. 

\begin{figure*}[t]
\begin{tabular}{@{\extracolsep{\fill}}c@{}c@{}c@{}c@{}@{\extracolsep{\fill}}}
\includegraphics[width=0.25\linewidth]{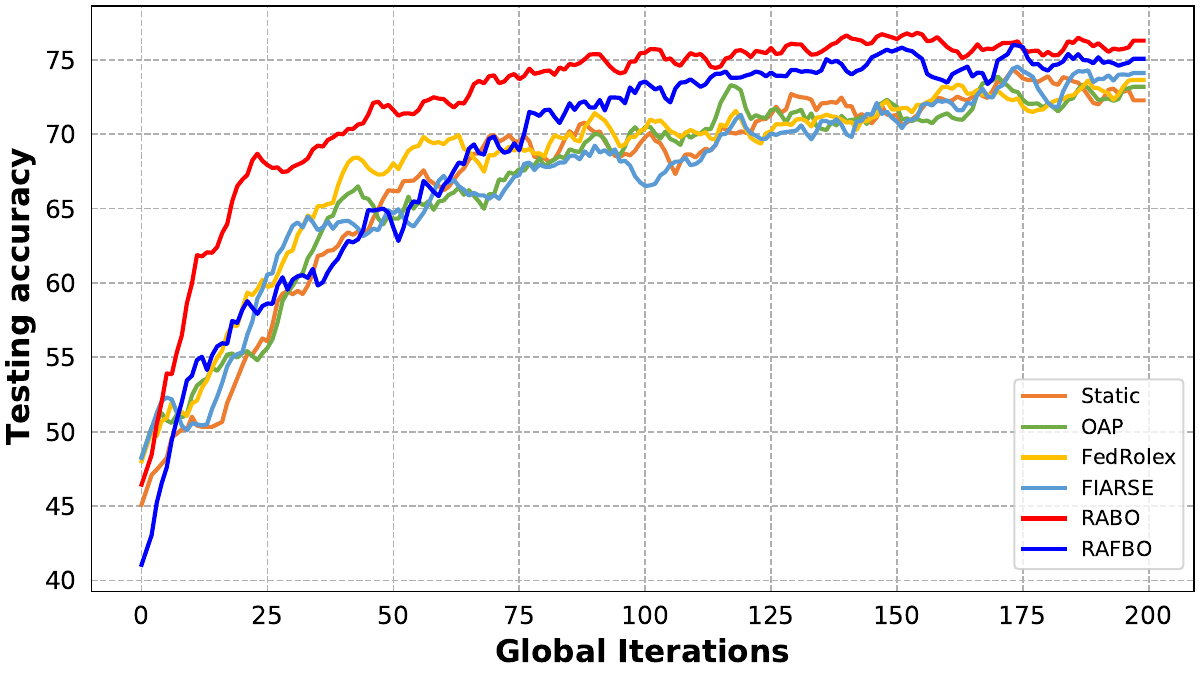}&
\includegraphics[width=0.25\linewidth]{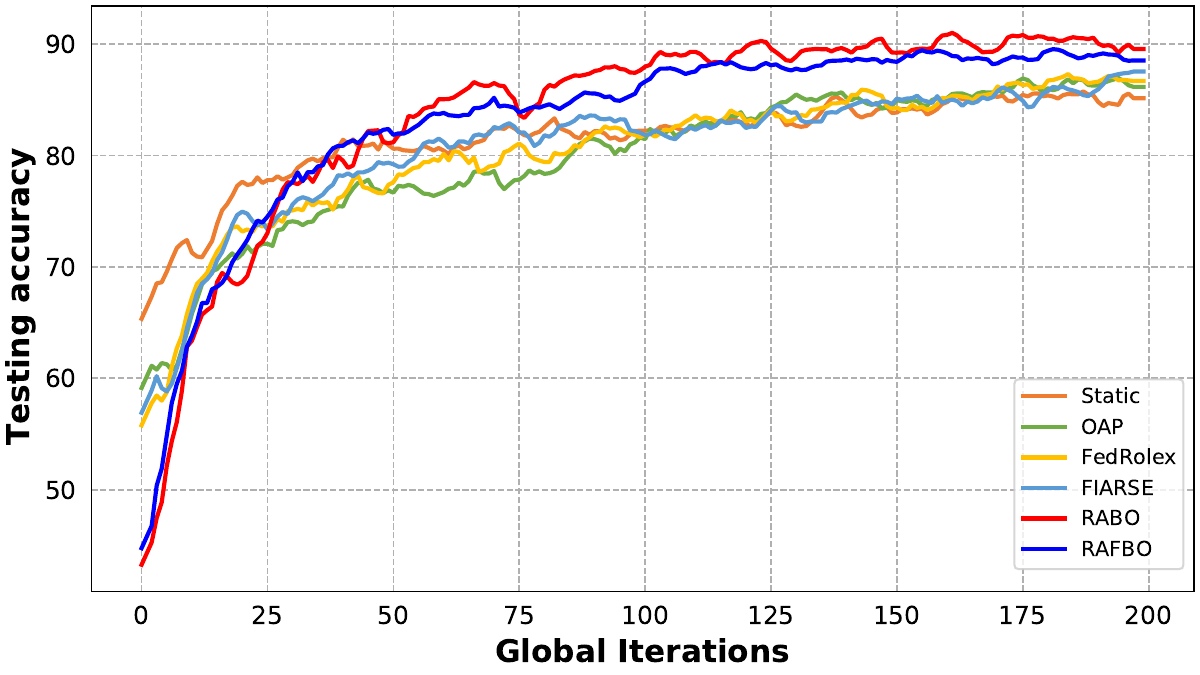}&
\includegraphics[width=0.25\linewidth]{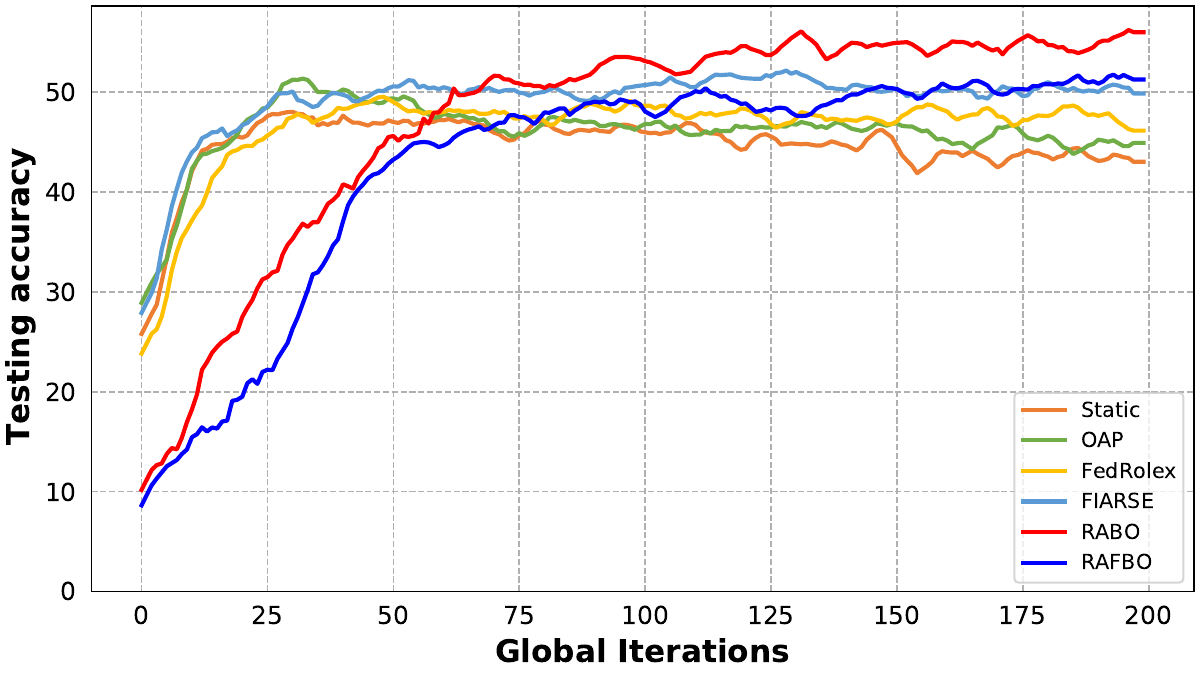}&
\includegraphics[width=0.25\linewidth]{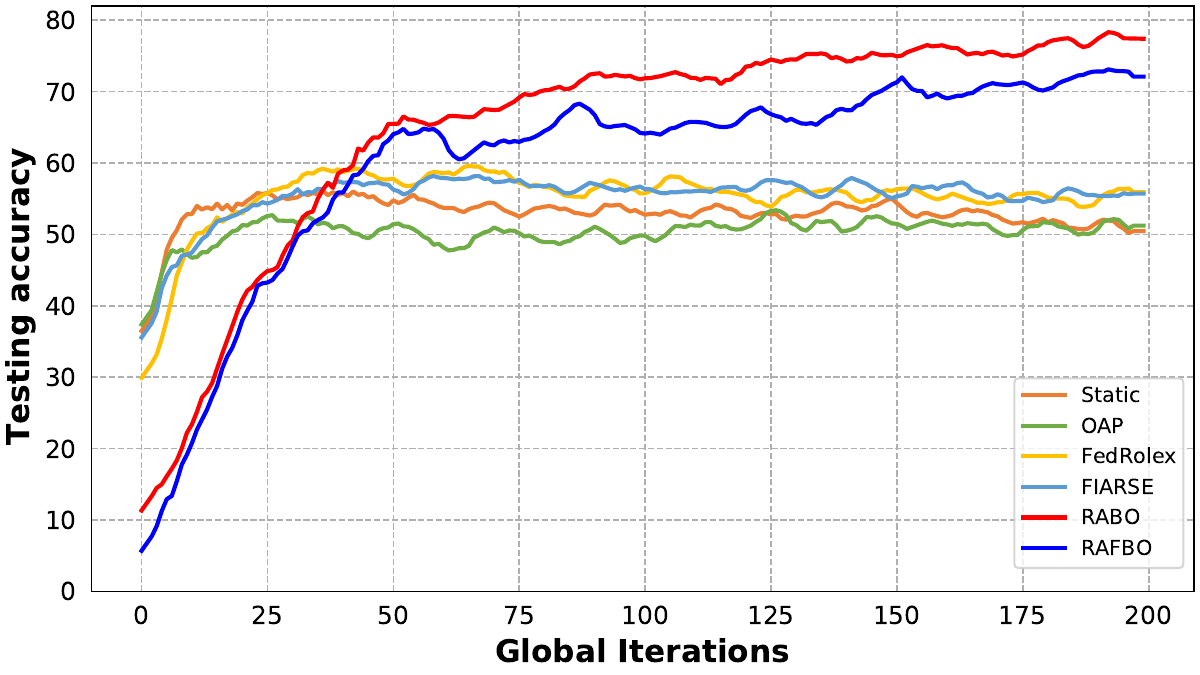}\\
(a) 5-way-1-shot & (b) 5-way-5-shot& (c) 20-way-1-shot& (d) 20-way-5-shot
\end{tabular}
{\caption{Comparison of different methods: the test accuracy v.s. \# global iterations in few-shot task}\label{fig:few_convergence}}
\end{figure*}

\textbf{Performance Evaluation.} The results on four N-way-K-shot settings are shown in Table~\ref{tab:baselines} and Figure~\ref{fig:few_convergence}. Overall, our proposed \textit{RABO} and \textit{RAFBO} achieves state-of-the-art performance compared to the baselines except for three full model training methods. \textit{RABO} improves 9.43\% on average and \textit{RAFBO} improves 6.59\%. Increasing N from 5 to 20 introduces more challenging tasks due to larger class diversity, our methods still perform better than other single-level submodel training methods. It's noticed that under 20-way-5-shot settings, \textit{RABO} achieves the competitive performance compared to distributed version of ProtoNet training with full model. Although the training model scale of \textit{RABO} is much smaller than it, it sufficiently shows the superior of bilevel optimization that learns a robust hyper-representation for fast adaptation. In Figure~\ref{fig:cost_few}, we further analyze the computation cost and communication cost compared to the FedMBO with full model. 
It can be seen that our proposed methods can achieve $5\times$ reduction in computation cost and $4\times$ reduction in communication cost separately. Specifically, the proposed \textit{RABO} accounts for 26.7\% of the computation cost, while \textit{RAFBO} accounts for 17.8\% when compared to FedMBO. Regarding the communication cost, both methods remarkably reduce it by 73.3\%.
These notable reductions can be attributed to the submodel training approach for both inner- and outer- parameters, where each client trains only a portion of the full model based on its capacity.

\textbf{Impact of $\mathcal{C}_x^*$.} To explore the impact of minimum covering number of outer parameters, we change the model capacity $\beta$ to $1/4$ that all clients only train $1/4$ outer parameters. And we manually control which segment each client trains to set $\mathcal{C}_x^*=\{1,2\}$. The results of 5-way-5-shot task are shown in Figure~\ref{fig:hyper}. We can observe that for both \textit{RABO} and \textit{RAFBO}, the result of $\mathcal{C}_x^*=2$ is better than $\mathcal{C}_x^*=1$. It means the more frequently the parameters are trained, the more knowledge from local data is integrated into the more part of global model, and leads to the faster convergence rate, which is identical to our theoretical analysis.
\begin{figure}[t]
\centering
\begin{minipage}{4cm}
\centering
\includegraphics[width=0.98\textwidth]{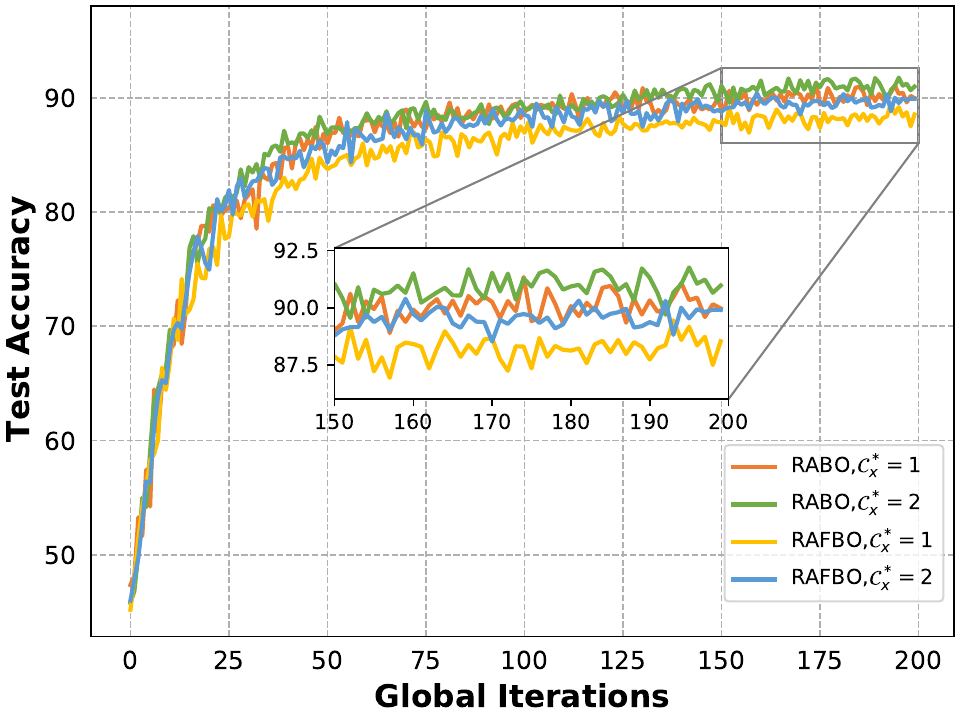}
\caption{The impact of different $\mathcal{C}^*_x$ in hyper-representation few-shot task.} 
\label{fig:hyper}
\end{minipage}
\begin{minipage}{4cm}
\centering
\includegraphics[width=0.98\textwidth]{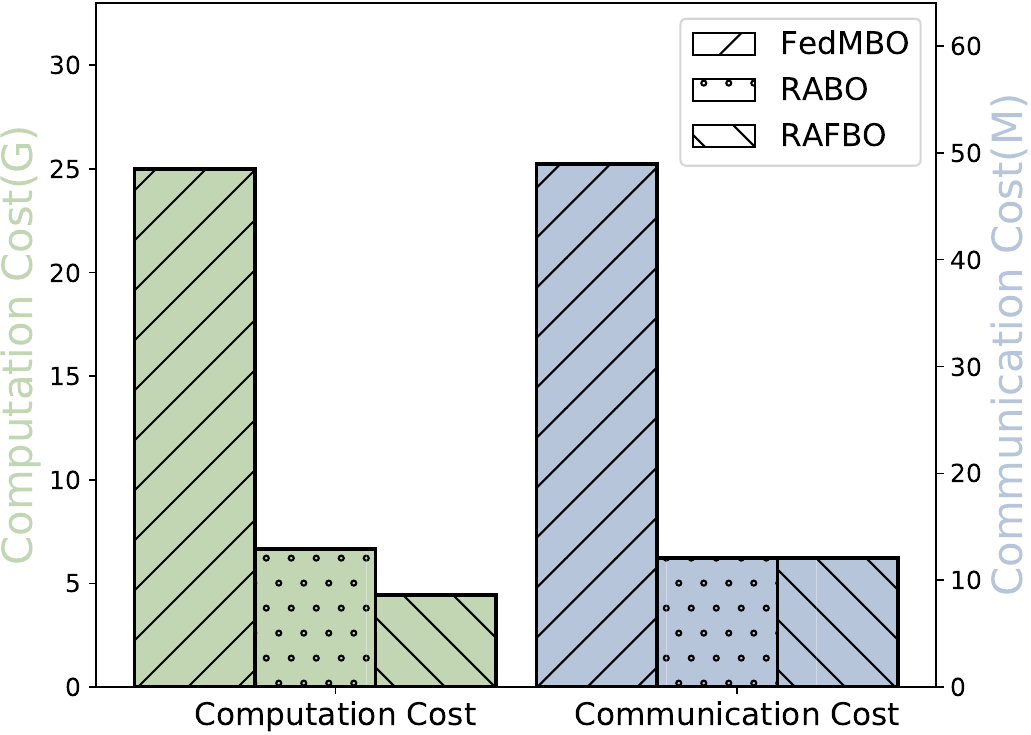}
\caption{Computation and communication cost comparison in few-shot task.} \label{fig:cost_few}
\end{minipage}
\end{figure}

\begin{figure}
\centering
\begin{minipage}{4cm}
\centering
\includegraphics[width=0.98\textwidth]{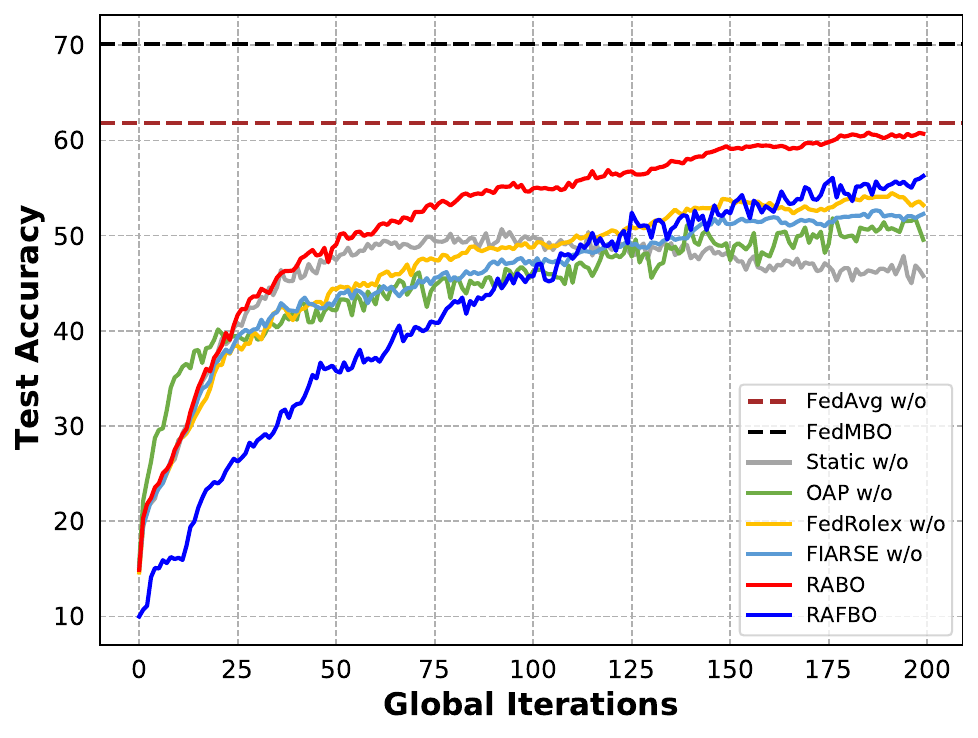}
\caption{Comparison of different methods in loss function tuning task.} \label{fig:convergence}
\end{minipage}
\begin{minipage}{4cm}
\centering
\includegraphics[width=0.98\textwidth]{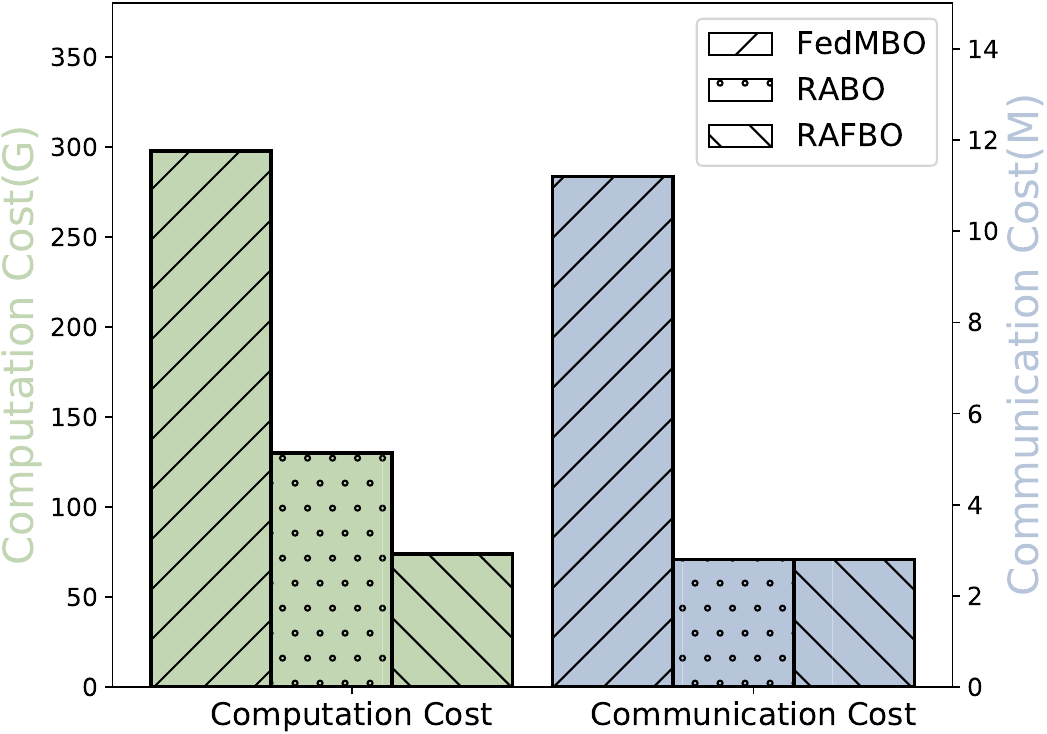}
\caption{Computation and communication cost comparison in loss function tuning task.} \label{fig:cost_tune}
\end{minipage}
\end{figure}

\subsection{Loss Function Tuning on Imbalanced Dataset}
In this section, we consider using the bilevel optimization to tune the hyper-parameters in the loss function for learning an imbalanced CIFAR-10~\cite{krizhevsky2009learning} dataset. We tune the VS-loss~\cite{kini2021label} in distributed setting as follows. VS loss is defined as $\mathcal{L}_{VS}(\textbf{x},y)=-w_k\log (e^{\Delta_kf_k(\alpha, y)+l_k}\big/\sum_{c\in[c]}e^{\Delta_cf_c(\alpha, y)+l_c})$, where $\textbf{x}=[\textbf{w}, \textbf{l}, \Delta]$ is the set of hyper parameters, $\alpha$ is the input data, $k$ is the label, $y$ is the parameter of model. Thus we can gain the optimization target as follows,
\begin{center}
\scalebox{0.93}{$
\begin{aligned}
&\min_{\textbf{x}} \Phi(\textbf{x})= 
\frac{1}{n} \sum_{i=1}^n\mathcal{L}_{VS}^{val}(\textbf{x},y^*(\textbf{x})) \nonumber
\\& \;\;\mbox{s.t.} \quad y^*(\textbf{x})\in \arg\min_{y} 
\frac{1}{n}\sum_{i=1}^n\mathcal{L}_{VS}^{tr}(\textbf{x},y).
\end{aligned}
$}
\end{center}
We create imbalanced training and validation datasets by following a long-tail distribution, where the number of examples per class decreases exponentially, that is $n_i'=n_i\mu^i$, where $n_i$ is the original number of examples for class i. We split it into 10 shards as local data. Unlike the few-shot learning task, we employ 80\% training data to train the inner model Resnet18 and 20\% validation data to fine tune the outer hyper parameters on each client. We compare the same baselines as in Section~\ref{sec:few-shot} and set all client capacities $\beta=1/2$.

\textbf{Performance Evaluation.} The overall performance is shown in Figure~\ref{fig:convergence}. We plot the test accuracy against global iterations to show the convergence and effectiveness of our proposed methods. The horizontal dashed lines serve as full model training baselines: black depicts accuracy reached
by distributed bilevel optimization algorithm FedMBO, and brown depicts accuracy in single-level optimization FedAVG without any loss tuning. Overall, our proposed \textit{RABO} and \textit{RAFBO} achieve state-of-the-art performance compared to baselines. \textit{RABO} improves 9.97\% on average and \textit{RAFBO} improves 5.53\%. In Figure~\ref{fig:cost_tune}, the brown dashed line is the accuracy without tuning the loss function and the black
dashed line is the accuracy with full model bilevel optimization. It's noticed that \textit{RABO} achieves a competitive performance compared to full model training FedAVG without loss tuning. Considering the computation cost, \textit{RABO} achieves $2\times$ reduction and \textit{RAFBO} achieves $4\times$ compared to FedMBO. For communication costs, both of them only take up 25\%. It demonstrates the computation and communication efficiency of our proposed methods.





\subsection{Scalability} In order to explore the scalability of our proposed methods \textit{RABO} and \textit{RAFBO}, we conduct the experiment that the number of clients is 10, 20, 30, 50 on \textit{few-shot learning} task. 
The results are shown in Table~\ref{tab:scalability}. We can find that as the number of clients increasing, the performance decreases. The main reason is the local amount of data samples is decreased, which further affect the performance of the local model. 

\begin{table}
\centering
\caption{Experimental results on scalability studies of proposed methods on few-shot task.}
\begin{tabular}{ccccc}
\toprule
\textbf{Method} & \textbf{10-Clients} & \textbf{20-Clients} & \textbf{30-Clients}& \textbf{50-Clients} \\ \midrule
\rowcolor{mypurple} \multicolumn{5}{c}{\textit{\textcolor{gray}{5-way-1-shot}}} \\
RABO & 76.05\% & 75.81\% & 74.39\%&72.76\% \\ 
RAFBO & 75.09\% & 74.85\% &73.04\%&71.80\% \\
\midrule
\rowcolor{mypurple} \multicolumn{5}{c}{\textit{\textcolor{gray}{5-way-5-shot}}} \\
RABO & 89.12 & 88.52 & 87.76 & 86.58 \\ 
RAFBO & 88.57 & 87.53 &87.03&86.44 \\
\bottomrule
\end{tabular}
\label{tab:scalability}
\end{table}

\section{Conclusion}
Distributed bilevel optimization faces the excessive computation challenge when applied in resource-limited clients. In this work, we designed a resource-adaptive efficient distributed bilevel optimization framework \textit{RABO} and based on it we proposed a new hypergradient estimator inspired by finite differences. We theoretically proved both of them can achieve an asymptotically optimal convergence rate $\mathcal{O}(1/\sqrt{\mathcal{C}^*_xQ})$ under the non-convex-strongly-convex condition. Extensive experiments on two tasks confirmed that the effectiveness and computation efficiency of proposed methods.

\bibliographystyle{plain}
\bibliography{mybib}
\begin{IEEEbiography}[{\includegraphics[width=1in,height=1.25in,clip,keepaspectratio]{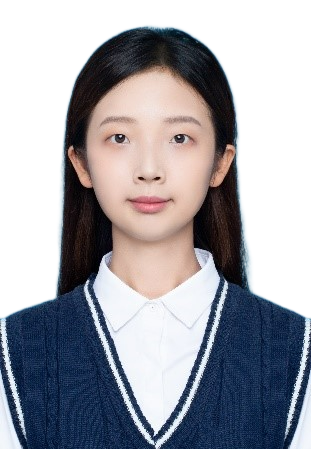}}]{Mingyi Li} is currently a Ph.D. student in the School of Computer Science and Technology, Shandong University. She received her B.S. degree in Shandong University. Her research interests include distributed collaborative learning and the theoretical optimization of distributed algorithms.
\end{IEEEbiography}

\begin{IEEEbiography}[{\includegraphics[width=1in,height=1.25in,clip,keepaspectratio]{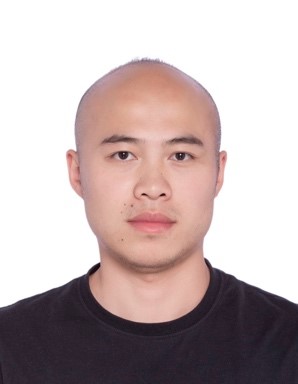}}]{Xiao Zhang} is now an associate professor in the School of Computer Science and Technology, Shandong University. His research interests include data mining, distributed  learning and federated learning. He has published more than 30 papers in the prestigious refereed journals and conference proceedings, such as IEEE Transactions on Knowledge and Data Engineering, IEEE Transactions on Mobile Computing, NeurIPS, SIGKDD, SIGIR, UBICOMP, INFOCOM, ACM MULTIMEDIA, IJCAI, and AAAI. 
\end{IEEEbiography}

\begin{IEEEbiography}[{\includegraphics[width=1in,height=1.25in,clip,keepaspectratio]{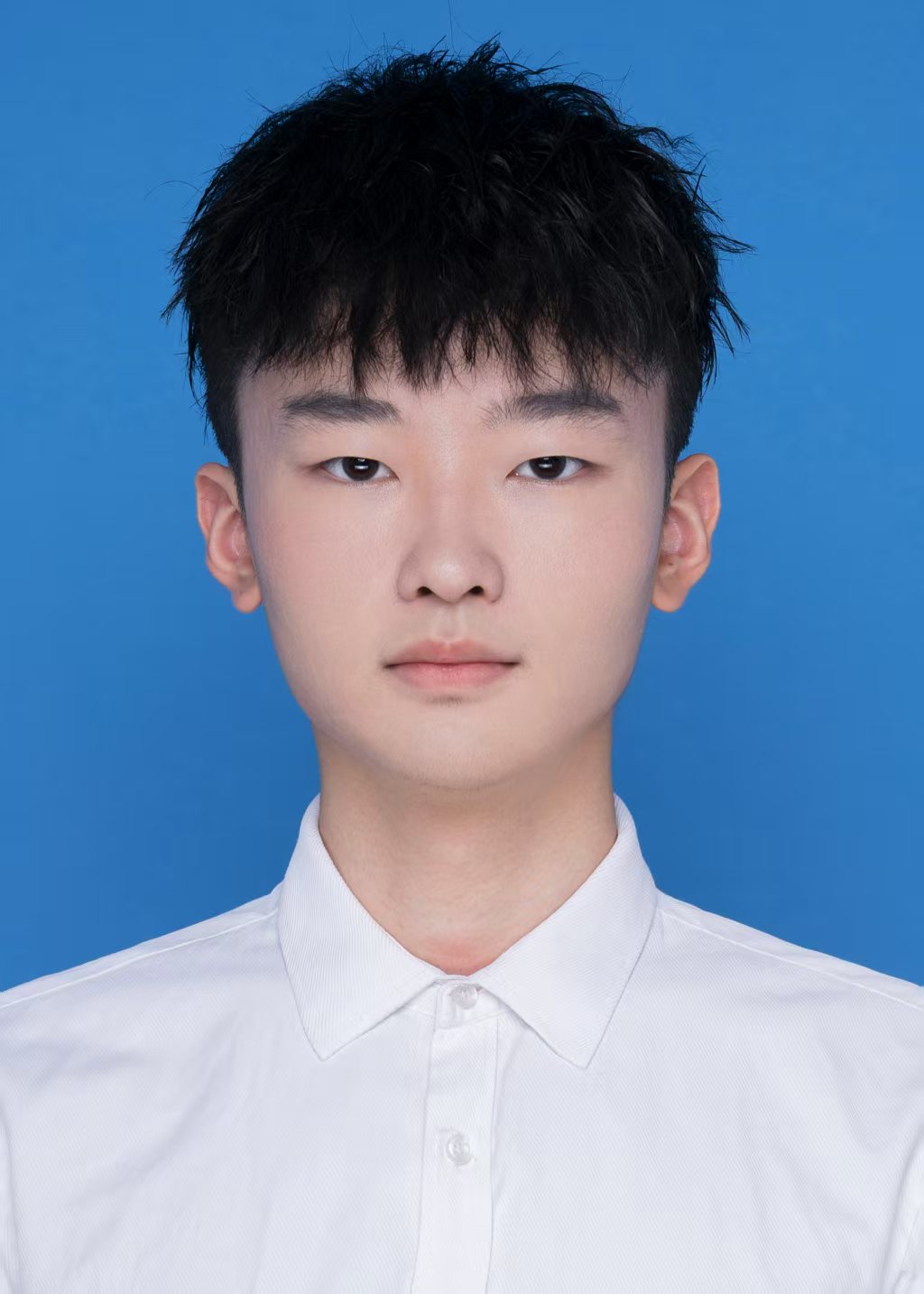}}] {Ruisheng Zheng} is currently an undergraduate student in the School of Cyber Science and Technology, Shandong University. His research interests include distributed learning and graph representation learning.
\end{IEEEbiography}

\begin{IEEEbiography}[{\includegraphics[width=1in,height=1.25in,clip,keepaspectratio]{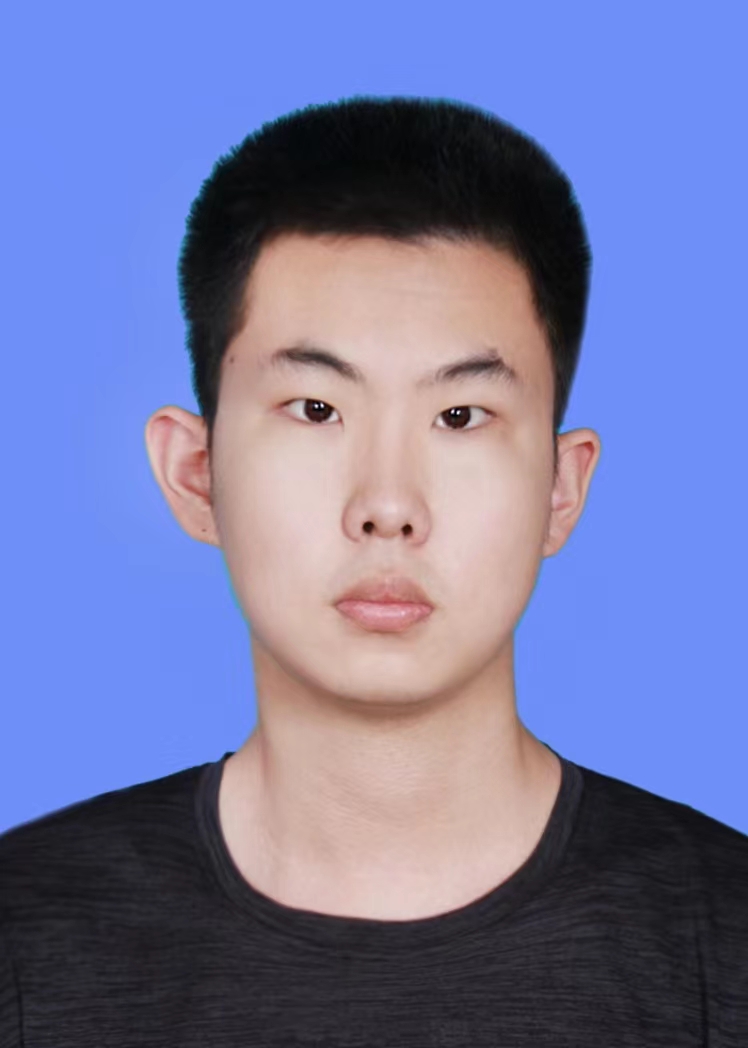}}]{Hongjian Shi} is currently an undergraduate student in the School of Computer Science and Technology, Shandong University. His research interests include distributed learning and intelligent computing network.
\end{IEEEbiography}

\begin{IEEEbiography}[{\includegraphics[width=1in,height=1.25in,clip,keepaspectratio]{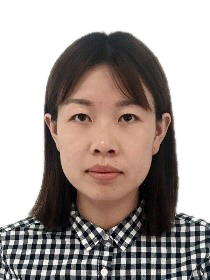}}]{Yuan Yuan} received the BSc degrees from the School of Mathematical Sciences, Shanxi University in 2016, and the Ph.D. degree from the School of  Computer Science and Technology, Shandong University, Qingdao, China, in 2021. She is currently a postdoctoral fellow at the Shandong University-Nanyang Technological University International Joint Research Institute on Artificial Intelligence, Shandong University. Her research interests include distributed computing and distributed machine learning.
\end{IEEEbiography}

\begin{IEEEbiography}[{\includegraphics[width=1in,height=1.25in,clip,keepaspectratio]{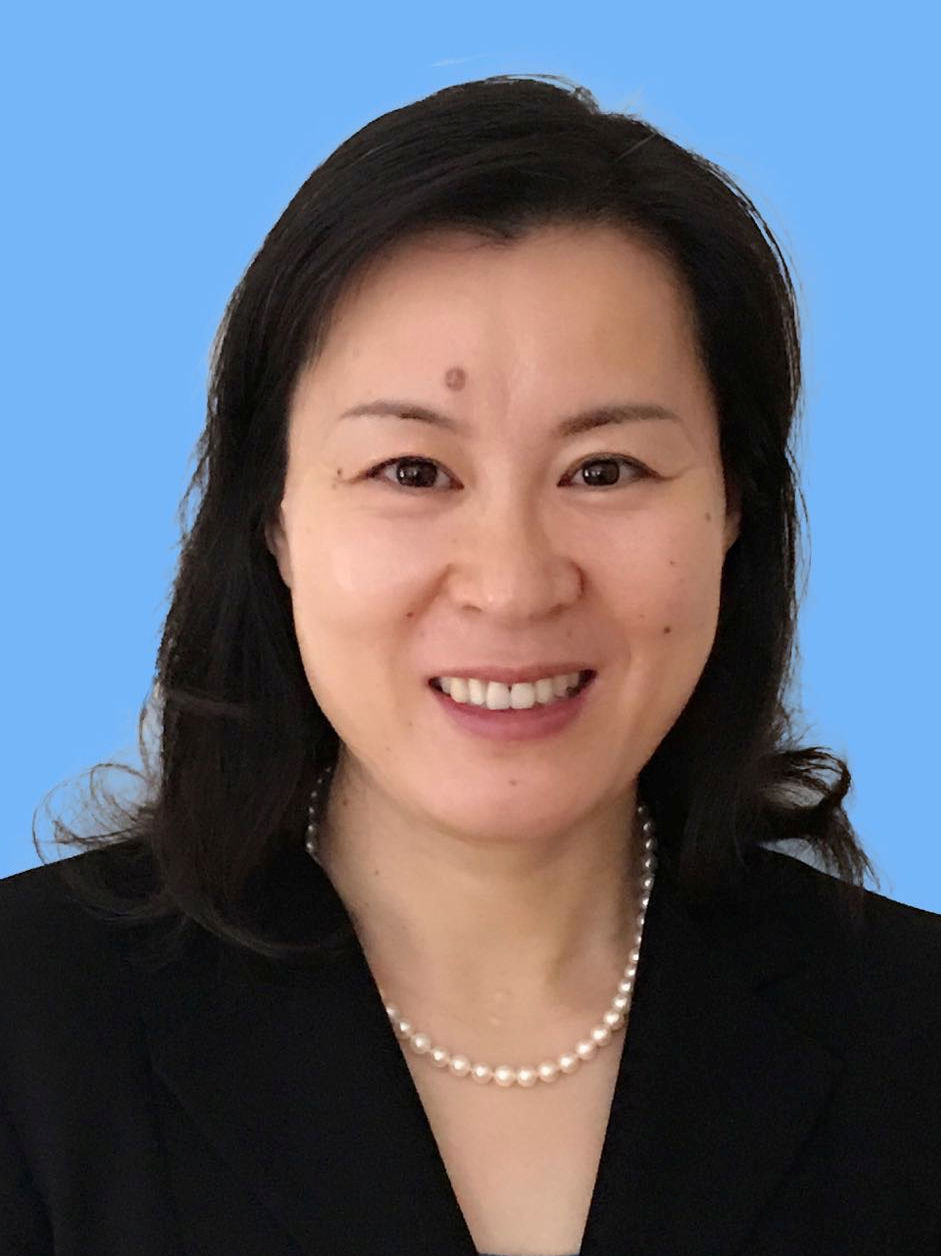}}]{Xiuzhen Cheng} (Fellow, IEEE) received the M.S. and Ph.D. degrees in computer science from the University of Minnesota-Twin Cities in 2000 and 2002, respectively. She is currently a Professor with the School of Computer Science and Technology, Shandong University. She worked as the Program Director of U.S. National Science Foundation (NSF) from April 2006 to October 2006 (full-time) and from April 2008 to May 2010 (part-time). Her current research interests include wireless and mobile security, cyber physical systems, wireless and mobile computing, sensor networking, and algorithm design and analysis. She is a member of ACM. She received the NSF CAREER Award in 2004. She also has chaired several international conferences. She has served on the editorial boards for several technical journals and the technical program committees for various professional conferences/workshops.
\end{IEEEbiography}

\begin{IEEEbiography}[{\includegraphics[width=1in,height=1.25in,clip,keepaspectratio]{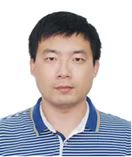}}]{Dongxiao Yu} received the B.S. degree in 2006 from the School of Mathematics, Shandong University and the Ph.D degree in 2014 from the Department of Computer Science, The University of Hong Kong. He became an associate professor in the School of Computer Science and Technology, Huazhong University of Science and Technology, in 2016. He is currently a professor in the School of Computer Science and Technology, Shandong University. His research interests include edge intelligence, distributed computing and data mining.
\end{IEEEbiography}
\end{document}